\documentclass[12pt]{article}
\usepackage{mathptmx}
\usepackage[utf8]{inputenc} %
\usepackage[T5,T1]{fontenc}

\usepackage{geometry}
\usepackage{color}
\usepackage{float}
\usepackage{mathtools}
\usepackage{amsmath}
\usepackage{amssymb}
\usepackage[retainorgcmds]{IEEEtrantools}
\usepackage{graphicx}

 \usepackage{natbib}
 \definecolor{upf}{RGB}{192,0,37}

\usepackage{subcaption}

 \usepackage[colorlinks=true,
            linkcolor=upf,
            urlcolor=upf,
            citecolor=upf]{hyperref}

\makeatletter

\usepackage{PlayfairDisplay} 
\usepackage{mathpazo}        

\usepackage{fullpage}

\usepackage{eurosym}
\usepackage{empheq}
\usepackage[normalem]{ulem}
\usepackage{color}
\usepackage{url}
\usepackage[section]{placeins}
\usepackage[doublespacing]{setspace}
\usepackage{booktabs}
\usepackage{caption}\usepackage{fixltx2e}\usepackage[flushleft]{threeparttable}
\usepackage[bottom]{footmisc}
\usepackage{array}
\usepackage{longtable}\usepackage{caption}
\usepackage[labelfont=bf]{caption}
\usepackage[toc,page]{appendix}
\usepackage[autostyle]{csquotes}
\usepackage{amsthm}

\newtheorem{theorem}{Theorem}[section]

\newtheorem{lem}[theorem]{Lemma}
\newtheorem{prop}{Proposition}

\newcounter{deferred}
\newcommand{\deferred}[2][]{%
  \ifstrempty{#1}
    {\stepcounter{deferred}\expandafter\gdef\csname temp\arabic{deferred}\endcsname{#2}}
    {\expandafter\gdef\csname temp#1\endcsname{#2}}%
}
\newcommand{\shownow}[1]{\csname temp#1\endcsname}

\makeatother

\begin{document}

\onehalfspace

\title{\textbf{{Sorting along Business Cycles}}\thanks{We thank seminar and conference participants at Dongbei University of Economics and Finance, and Liaoning University.  Gola gratefully acknowledges the support of the UKRI Horizon Europe Guarantee (Grant Ref: EP/Y028074/1).  The usual disclaim applies.}\textbf{\ }}

	\author{Paweł Gola\thanks{University of Edinburgh, 
 \protect\href{mailto:pawel.gola at ed.ac.uk}{pawel.gola at ed.ac.uk}}
		\and Haozhou Tang\thanks{Dongbei University of Finance and Economics, \protect\href{mailto:haozhoutang\%dufe.edu.cn}{haozhoutang@dufe.edu.cn}. }}

	\date{\today}
 \maketitle   

\begin{abstract}
We develop an analytically tractable model featuring heterogeneous workers and firms, where labor markets clear through a one-to-many sorting mechanism. Firms determine both the number and composition of their employees, shaping (1) the income distribution among workers and (2) the productivity distribution across firms. We study business cycles driven by market efficiency shocks that disproportionately benefit more productive firms. The model's implications are consistent with empirical regularities on the cyclical behavior of wage and productivity distributions.

 \end{abstract}



\newpage

\section{Introduction}

A prevailing convention in business-cycle studies is to treat firm-level productivity as exogenous, determined outside the firm's recruitment decisions. Firms are typically modeled as passive takers of their productivity draws, making hiring and investment decisions conditional on those primitives.  Yet this view abstracts from a fundamental aspect of how firms actually operate: productivity is shaped by whom firms hire. A firm's capabilities hinge on the talents it brings in---whether a CEO who sets strategic direction, a CFO who manages financial structure, R\&D personnel who drive innovation, or sales teams who determine market reach. In many cases, hiring decisions are not responses to productivity; they are determinants of it. 
This perspective highlights that productivity is an equilibrium outcome of how heterogeneous firms and workers sort in the labor market.

Recognizing productivity as an equilibrium sorting outcome raises natural questions: how are the distribution of wages and firm-level productivity connected, and how do they evolve over the business cycle?
A salient feature of the productivity distribution is that its dispersion is countercyclical---widening in recessions and narrowing in booms.\footnote{See, for example, \cite{bachmann2014}, \cite{kehrig2015}, \cite{bloom2018}, and \cite{cunningham2023}.}  This variation is typically taken as exogenous, often attributed to uncertainty or volatility shocks.  
In contrast, wage inequality, measured as the variance of log wages, exhibits a distinct cyclical pattern.  As shown in Figure \ref{fig:wage inequality}, the left panel reveals a long-run upward trend in wage dispersion from \cite{song2019}, along with clear cyclical movements: wage inequality tends to rise in economic expansions and contract in recessions, highlighted by the NBER-dated shaded regions. The right panel, which plots detrended wage inequality and detrended GDP, makes the comovement more transparent---inequality tends to track the business cycle. Notably, the patterns differ across episodes. For instance, in the late 1990s, this comovement was less pronounced: wage inequality remained below its own trend despite the growth of real GDP. By contrast, wage inequality dropped sharply during the Great Recession.\footnote{Throughout this paper, we focus on wage inequality rather than income inequality, which includes all sources of income. The figure uses wage income, i.e., labor earnings, and thus only accounts for employed individuals, whereas income inequality also includes the unemployed---potentially leading to very different patterns. Additionally, different definitions and measures of income inequality can produce substantially different results. See, for example, \cite{piketty2003} and \cite{heathcote2020}.} 

Overall, the data suggest a potential contrast between the two dispersions: while the dispersion of wages---reflecting worker performance---tends to increase in expansions, the dispersion of firm-level productivity---reflecting firm performance---tends to rise in downturns. 
These patterns
highlight an important feature for understanding firm-worker interactions over the business cycle.

\begin{figure}[H]
    \centering
    \includegraphics[
        width=1\textwidth,
        trim=0cm 9.8cm 0cm 9.5cm,
        clip
    ]{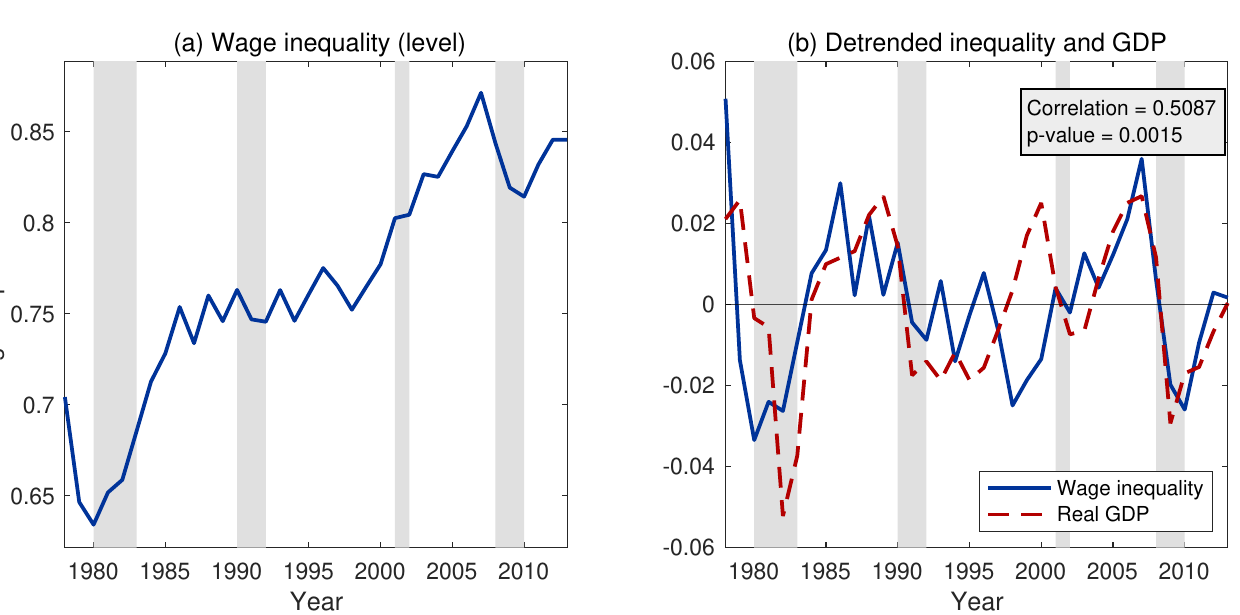}
    \caption{Wage inequality over business cycles}
    \caption*{\footnotesize\emph{Notes:} Shaded areas indicate NBER recessions. 
    Wage inequality is measured by variance of individual log wages. 
    Both series in panel (b) are detrended.\\  
    \emph{Source:} \cite{song2019} and authors' calculations.}
\label{fig:wage inequality}
\end{figure}

Motivated by these patterns, we introduce a framework with heterogeneous firms and workers in which labor markets clear through a one-to-many matching mechanism. Firms differ in their underlying types, and workers differ in their skill levels. Firms choose both the number and type of their employees, shaping the distributions of wages and productivity, which exhibit Pareto tails in equilibrium. Apart from the introduction of worker heterogeneity and one-to-many sorting, the model is entirely standard, featuring CES demand and decreasing-returns Cobb-Douglas production technology. Crucially, productivity emerges from the equilibrium sorting between firms and workers rather than being imposed exogenously. In this sense, the model generalizes the standard heterogeneous-firm framework: without worker heterogeneity or sorting, it reduces to the canonical setting where productivity is given.

Sorting between firms and workers is sensitive to aggregate conditions: changes in the macro environment can shift the composition of matches and thus alter the distributions of wages and productivity. We introduce a reduced-form market efficiency shock that operates through time-varying firm-specific wedges on labor and capital prices, representing distortions that fluctuate over the business cycle. This reduced-form shock may capture the compound effects of several underlying structural shocks and disproportionately affects high-type firms. It is designed to generate business-cycle fluctuations that reproduce procyclical reallocation of labor to more advanced firms, a key cyclical pattern observed in the data.\footnote{See, for example, \cite{moscarini2012}, \cite{haltiwanger2018}, \cite{crane2023}, and \cite{haltiwanger2025}.}

We analytically characterize how the reduced-form market efficiency shocks shape the cross-sectional distributions of wages and firm-level productivity. A positive shock induces high-type firms to expand more aggressively, shifting the composition of jobs toward higher-quality positions and thereby widening wage dispersion. Simultaneously, due to this relative expansion, these firms are matched with relatively lower-type workers, pulling their realized productivity closer to that of low-type firms, which shrink in relative size, and thus dampening the dispersion of firm-level productivity.  This mechanism demonstrates how shifts in the allocation of talent across firms can generate opposing cyclical movements in wages and productivity, even when underlying technologies are unchanged.
It is consistent with the findings in \cite{crane2023},  who document that during downturns high-rank workers are more likely to work at low-rank firms, while low-rank workers are less likely to be employed by high-rank firms.

While the mechanism described above can generate opposing cyclical movements in wage and productivity dispersions, not necessarily all recessions exhibit this disparity. To understand the broader range of patterns, we also examine the effects of several structural shocks. In particular, we consider aggregate productivity shocks and second-moment shocks that alter the variance of firm types and firm-level distortions. This allows us to analyze how different shocks shape the joint dynamics of wage and productivity dispersions over the business cycle.

\paragraph{Literature review}

The main conceptual contribution of this work is the finding that, under reasonable assumptions, changes in the matching function have opposing effects on the distributions of wages and measured productivity. In particular, the result that improvements in workers' matches decrease the inequality of measured productivity is, to the best of our knowledge, completely novel. Combined with the well-known fact that, in a model with worker-firm type complementarities, improvements in workers matches increase wage inequality \citep{Costinot2010}, this implies our main result: cyclical fluctuations in sorting cause the procyclical fluctuations of wage inequality and the counter-cyclical behavior of productivity dispersion.

Our contributions to the literature on labour market sorting extend beyond these novel comparative statics. First, while frictionless one-to-many sorting models (see \cite{eeckhout2018} for the most comprehensive theoretical treatment) have been used to study a wide array of applied questions in labour economics \citep[e.g.,][]{Sattinger1975, Teulings1995, Teulings2005} and international trade \citep[e.g.,][]{Costinot2009, Costinot2010, Sampson2014, grossman2017, Choi2023}, we are the first to embed this framework into an otherwise canonical business cycle model with aggregate uncertainty and capital accumulation. Second, we develop a novel, tractable specification of the one-to-many sorting model, which admits a fully analytical solution to the sorting problem. While this specification draws inspiration from \cite{gola2023}, notably by adopting a similar production function, it improves on their approach considerably by allowing for exponentially distributed firm and worker types, which is precisely what produces Pareto tails of wage, productivity, and firm size distributions in equilibrium.

Finally, our paper contributes to the literature on heterogeneous firms in several ways. First, unlike canonical models that treat firm-level productivity as exogenous, we study the endogeneity of productivity through equilibrium sorting between heterogeneous firms and workers. This mechanism generates a rich mapping from firm-worker matches to the cross-sectional distributions of productivity and wages. In particular, it offers a novel interpretation of uncertainty or volatility shocks. While pioneering studies such as \cite{bloom2009}, \cite{bachmann2013, bachmann2014}, \cite{schaal2017}, \cite{bloom2018}, and \cite{arellano2019} introduce these shocks into heterogeneous-firm frameworks as exogenous drivers of fluctuations in productivity dispersion and aggregate dynamics, we show they can be partially reconciled as an endogenous outcome of time-varying worker-firm sorting.
Second, by incorporating worker heterogeneity, we extend the analysis to wage inequality and its cyclical behavior. Finally, the model remains highly tractable, allowing for analytical characterization of a broad range of cross-sectional and business-cycle implications.

The paper is organized as follows. Section \ref{sec:model} introduces the model, and Section \ref{sec:eqm} presents the equilibrium analysis. In Section \ref{sec:micro}, we show how aggregate conditions shape the cross-sectional distributions of wages and productivity. Section \ref{sec:4.3} discusses alternative drivers of business cycles, and Section \ref{sec:conclude} concludes.

\section{Model}\label{sec:model}
In this section, we develop a model that incorporates one-to-many sorting between heterogeneous firms and workers. Firms make decisions regarding both the quantity of labor, i.e., the number of workers to hire, and the quality of labor, i.e., the skill level or type of workers. In this framework, firm-level productivity is determined by not only firms' own type, but also their match with employee types.  
We demonstrate that our model nests a standard heterogeneous-firm model with exogenous idiosyncratic productivity.

Each firm is subject to firm-specific distortions, represented by labor and capital wedges. The distortions are i.i.d. across firms but are correlated with firm types. Dynamically, we introduce shocks to market allocative efficiency, modeled as a time-varying correlation between firm types and firm-specific distortions. 
\subsection{Household}
Time is discrete and runs from zero to infinity.  
A representative household consists of a unit continuum of members, indexed by $i\in[0,1]$, each endowed with one unit of labor and skill type follows exponential distribution $x_i\sim \text{Exp}\left(\lambda_x\right)$. The household maximizes lifetime utility function
$$
\max E_0 \sum_{t=0}^{\infty} \beta^t \log C_t
$$
where $C_t$ denotes the aggregate consumption, $\beta$ represents the discount factor.

The household owns all firms, as well as capital and labor inputs. It is subject to the following budget constraint:
\begin{align}
 C_t+K_{t+1}-(1-\delta)K_t=\int_0^1w_{it}di+R_tK_t+D_t,   
\end{align}
where $K_t$ denotes the aggregate capital stock, $\delta$ represents its depreciation rate, $w_{it}$ denotes the wage received by household member $i$, $R_t$ represents the capital price, and $D_t$ represents operation profits from firms. Each family member works independently to earn wages, but their income is pooled together when making decisions about consumption and savings.

\subsection{Production}

\paragraph{Final good producer} There is a competitive sector of final good producers, all of which produce the homogeneous numeraire good  by combining differentiated intermediate goods $j\in[0,1]$
$$
Y_t =\left(\int_0^1 Y_{jt}^{1-\frac{1}{\xi}} d j\right)^{\frac{\xi}{\xi-1}}
$$
The price of final good is normalized to one.

\paragraph{Intermediate good producers} There is a unit continuum of intermediate firms. 
Production function follows\footnote{For notation ease, we omit time and firm-specific subscripts from the functional inputs in the following analysis.}
\begin{align}
Q_t(l, k, x, \theta)=A_t q(x, \theta) k^{\alpha}  l^\gamma,\label{production}
\end{align}
where $\theta$ and $x$ respectively denote the type of the firm and the type of worker that is matched with the firm, $k$ and $l$ represent the capital and labor input.  Firm type $\theta$ follows an exponential distribution $\theta\sim \mathrm{Exp}\left(\lambda_\theta\right)$. We assume that $\alpha+\gamma<1$. $A_t$ denotes a common productivity component that is identical across every firm. The idiosyncratic productivity component $q(x, \theta)$ follows
\begin{align}
    q(x, \theta) =\exp\left\{x^\psi\theta^{1-\psi}\right\}, \label{prod2}
\end{align}
where $\psi\in[0,1]$. Now, productivity depends not only on the firm's own type $\theta$ but also on the type of workers $x$ it employs. As discussed below, aggregate conditions affect the firm's hiring decisions, which in turn shape productivity.

Later we will discuss the economic interpretations for each of the functional parameters. For now, it is worth noting that a standard Cobb-Doulas function can be viewed as a special case of production function (\ref{production}). Specifically, if $\psi=0$, idiosyncratic productivity  is equal to $e^{\theta}$, independent of worker types. In this case, the production function (\ref{production}) is mathematically equivalent to a standard Cobb-Douglas function with decreasing-returns-to-scale (henceforth DRS): $e^{\theta}k^\alpha l^\gamma$, where $e^{\theta}$ can be interpreted as firm-level idiosyncratic productivity and is itself Pareto distributed.  Work type $x$ does not matter for the production, and all workers receive the same wage: no wage heterogeneity, just as in a standard heterogeneous-firm model.

\paragraph{Distortions and market efficiency shcoks}
Intermediate producers face demand schedule
\begin{align}
     Y_{jt}=(P_{jt})^{-\xi}Y_t. \label{demand}
\end{align}
At time $t$, firms compete in the labor market by offering wage $w_t(x)$ to workers of type $x$. We follow \cite{restuccia08} and \cite{hsieh09} by incorporating market frictions in the form of reduced-form, firm-specific wedges in capital and labor prices. These wedges introduce distortions that capture inefficiencies in resource allocation across firms. Given demand $y$, firms minimize cost by solving the following Lagrangian problem
\begin{align}
\min _{l, k,m} \tau_1  w_t(m)l+\tau_2 R_tk-\chi\left[Q_t(l, k, m, \theta)-y\right].\label{cost}
\end{align}
The firm-specific wedges $\tau_1$ and $\tau_2$ represent distortions arising from market frictions.  Lagrangian multiplier $\chi$ denotes the marginal cost.

Throughout the paper, we focus on labor market frictions and fluctuations in labor market efficiency. Specifically, we assume that labor wedge is correlated with firm type and follows
\begin{align}
    \tau_1=\exp(z_t\theta+\epsilon_1).\label{wedge}
\end{align}
where $z_t$ is constant across firms, while $\epsilon_1\sim N\left(0,\sigma_1^{2}\right)$ represents a firm-specific component that is i.i.d. across firms.
Thus, $z$ reflects the relationship between the (log of) distortion $\tau_1$  and firm type $\theta$. In terms of capital market, we assume that capital wedge follows $\tau_2=\exp(\epsilon_2)$, where $\epsilon_2\sim N\left(0,\sigma_2^{2}\right)$ is i.i.d. across firms. For simplicity, we omit the correlation between capital wedges and firm types. However, the analysis can be readily extended to include type-dependent capital market distortions by modeling capital wedges, similar to (\ref{wedge}), as being correlated with firm types.

Throughout the paper, we assume that $z_t>0$, $\forall t$. This assumption has two considerations. First, as discussed in Section \ref{sec:focs}, it implies that larger firms tend to allocate a smaller share of their revenue to wages, a pattern that aligns with well-documented empirical findings.\footnote{See \cite{autor2020}.} Second, this assumption is consistent with empirical evidence on the correlation between distortions and firm size and productivity, which shows that high-productivity firms tend to face higher marginal input costs and distortions.\footnote{See, for example, \cite{hsieh2014}, \cite{gourio2014}, \cite{garicano2016}, and \cite{bento2017}.} In particular, \cite{bento2017} estimate the elasticity of distortions with respect to productivity across countries and find values ranging from $0.22$ to $0.74$.

We follow \cite{chari2007} by introducing time-varying labor wedges into the model. Specifically, we assume  $z_t$ follows a Markov process. Fluctuations in $z_t$ affect firms asymmetrically depending on their type $\theta$, with higher-$\theta$ firms being more exposed to these shocks. We intentionally refrain from microfounding this reduced-form shock; instead, one may interpret $z_t$ as capturing the combined effect  of multiple underlying structural forces that vary over the business cycle.

This modeling choice is consistent with empirical evidence showing that labor reallocation toward more productive firms is procyclical. Recent studies by \cite{crane2023} and \cite{haltiwanger2025}, using a variety of firm-level productivity measures, document that during expansions high-type firms expand relative to low-type firms by recruiting workers who would otherwise be employed by low-type firms.  Firm size can also serve as a proxy for productivity: \cite{moscarini2012} documents that employment growth in large firms is more strongly correlated with aggregate economic activity than growth in smaller firms. Furthermore, some evidence suggests that these larger, more productive firms are more sensitive to structural shocks, such as monetary policy shocks.\footnote{See, for example, \cite{ottonello2020}, \cite{morlacco2021}, and \cite{kroen2021}. However, this evidence is subject to debate, as some studies document opposite effects; for instance, \cite{gertler1994} and \cite{crouzet2020}.} Finally, \cite{buera2015} develop models showing that credit crunches can be manifested as time-varying input wedges and cause the labor share of the most productive firms to shrink.

Note that the purpose of this paper is not to identify the structural sources of aggregate fluctuations. Rather, we examine what happens when better firms expand more during booms---specifically, how such asymmetric responses shape the distribution of wages and firm-level outcomes over the business cycle.

\section{Equilibrium Analysis} \label{sec:eqm}
An equilibrium consists of: (1) a consumption (and saving) path that maximizes the representative household's lifetime utility and satisfies its intertemporal budget constraint;  2) optimal capital inputs $k_{jt}$, labor quantities $l_{jt}$, and labor types $x_{jt}$ that maximize each firm $j$'s profits;  and (3) prices for intermediate goods $P_{jt}$, the rental rate of capital, $R_t$, and wage schedule for each skill type $w_t(x)$ that clear markets for all $t$.  

To sum up, the equilibrium definition above is analogous to that in a standard stochastic neoclassical growth model with heterogeneous firms, except that we incorporate one-to-many sorting: firms' recruiting decisions now involve both the quantity and the quality of labor. Labor market clearing requires that the supply of each skill type equals its demand.

This section presents the equilibrium analysis. We defer the detailed analysis of cross-sectional distributions to Section \ref{sec:micro}.
\subsection{One-to-many Sorting} 
We first define a \textit{job offer} as $h = \theta$. The type of job offers is determined by the underlying firm type $\theta$. However, we introduce a separate notation $h$ to distinguish job offers from firm types, since the distribution of job offers generally differs from that of $\theta$. In particular, the distribution of $h$ also depends on firm size:
\begin{align*}
    f_t(h)=\int_{-\infty}^{\infty}\int_{-\infty}^{\infty} l_t(h,\epsilon_1,\epsilon_2)*\lambda_\theta\exp(-\lambda_\theta h)\mathrm{d} \Phi(\frac{\epsilon_1}{\sigma_1}) \mathrm{d} \Phi(\frac{\epsilon_2}{\sigma_2})
\end{align*}
where $f_t(h)$ denotes the probability density function (PDF) of $h$, and $l_t(\theta,\epsilon_1,\epsilon_2)$ represents firm size. It is immediate that $h$ and $\theta$ may follow different distributions due to the incorporation of firm size. 

Our framework thus features a \textit{one-to-many} sorting structure, in which firms endogenously determine their size and post multiple job offers, in contrast to a one-to-one sorting setup where each firm corresponds to a single job. To see this more clearly, consider a conventional one-to-one sorting framework in which all firms have equal unit size, i.e., $l_t(\theta,\epsilon_1,\epsilon_2) = 1$. In that case, we have $f_t(h) = \lambda_\theta e^{-\lambda_\theta h}$, implying that $h$ and $\theta$ follow the same distribution.

To facilitate the analysis, we conjecture that
\[
f_t(h)=\lambda_t\exp(-\lambda_th),
\]
that is, the number of job offers follow an exponential distribution with a time-varying parameter $\lambda_t$. Throughout most of this section, we take this as given. In Section \ref{sec:3.3},
we verify this conjecture and characterize the equilibrium value of $\lambda_t$.

The function $f_t(h)$ can also be interpreted as the \textit{labor demand} generated by firms of type $\theta = h$. On the other hand, the supply of type-$x$ workers is assumed to follow an exponential distribution, $x_i \sim \text{Exp}(\lambda_x)$. 
The corresponding labor supply function is
\[
S(x)\equiv-\lambda_x\exp(-\lambda_xx).
\]
Let $h = \mu_t(x)$ denote the matching function between a worker of type $x$ and a job offer $h$. 
Labor market clearing requires that $\forall x\in[0,\infty)$
\begin{align}
S(x)=D(x)\equiv-\lambda_t\mu_t'(x)\exp(-\lambda_t\mu_t(x)).\label{labormktclear}
\end{align}
From (\ref{labormktclear}) and that $\int_0^\infty D_t(x)dx=1$, we obtain
\begin{align}
\mu_t(x) = \frac{\lambda_x}{\lambda_t}x. \label{pam}
\end{align}
The matching function (\ref{pam}) reveals a positive assortative matching (PAM) pattern between workers and job offers (and thus firms): workers with higher $x$ are matched with firms offering higher $h$ (and equivalently higher $\theta$).

\subsection{Firms' problem}\label{sec:focs}
The firm level state variables can be summarized by $(\theta,\epsilon_1,\epsilon_2)$. Firm-level variables can be expressed as functions of these state variables. Demand schedule (\ref{demand}) can be rewritten into
\begin{equation}
     Q_{t}(\theta,\epsilon_1,\epsilon_2)=(P_t(\theta,\epsilon_1,\epsilon_2))^{-\xi}Y_t, 
\end{equation}
where the price is a constant markup over the marginal cost
\begin{align}
P_t(\theta,\epsilon_1,\epsilon_2)=\frac{\xi}{\xi-1}\chi_{t}(\theta,\epsilon_1,\epsilon_2).
\end{align}

Next we proceed to analyze the input choices of intermediate firms.

\paragraph{First-order conditions and wage schedule}
As to intermediate firms, the firm's first-order condition (FOC) with respect to labor input $l$ is given by
\begin{align}
   \chi A_t\gamma k^{\alpha}l^{\gamma-1}\exp(x^{\psi} h^{1-\psi})=\tau_1w(x), \label{eq: lFOC}
\end{align}
where $\chi$ denotes the associated marginal cost. Note that equation (\ref{eq: lFOC}) implies that the labor share of a company is equal to $\frac{\xi-1}{\xi\tau_1}$, which is \textit{negatively} correlated with the type of firm $\theta$.

Similarly, the FOC with respect to worker type $x$ yields
\begin{align}
   \chi A_t \psi k^{\alpha}l^{\gamma}\exp(x^{\psi} h^{1-\psi})x^{\psi-1} h^{1-\psi}=\tau_1w'_t(x)l.
\end{align}
Dividing the second condition by \eqref{eq: lFOC} gives 
\begin{align}
    \frac{w'_t(x)}{w_t(x)}=\frac{\psi}{\gamma}x^{\psi-1} h^{1-\psi},\label{eq: wagederiv}
\end{align}
Applying the fundamental theorem of calculus, we obtain
\begin{align}
w_t(x)=w_t(0)\exp(\frac{\psi}{\gamma}(\frac{\lambda_x}{\lambda_t})^{1-\psi}x)
\end{align}

Taking logarithms gives
 \begin{align}\label{eq: wagex}
\ln(w_t(x))=\frac{\psi}{\gamma}(\frac{\lambda_x}{\lambda_t})^{1-\psi}x+\ln(w_t(0)),
\end{align}
a result that reveals log wages follow an exponential distribution---or, equivalently, that raw wage levels conform to a Pareto distribution.

Equation (\ref{eq: wagex}) characterizes how worker types are mapped to their wages. Workers with higher type $\theta$ consistently receive higher wages.  The steepness of this wage-worker type relationship hinges on two components, each capturing distinct economic tradeoffs.

First, the gradient increases with the ratio $\frac{\psi}{\gamma}$, a ratio determined by production technology. Here, $\psi$ quantifies the \textit{skill intensity} of worker contributions to firm productivity. Whereas $\gamma$ reflects \textit{scale intensity}, i.e. gains from expanding the workforce.  Their ratio encapsulates the core tradeoff firms face: prioritizing worker quality (skill) versus quantity (headcount). A high $\frac{\psi}{\gamma}$ ratio signals that quality dominates this tradeoff, forcing firms to compete more fiercely for high-type workers by offering a steeper wage premium tied to worker ability.

Second, the gradient is increasing in $(\frac{\lambda_x}{\lambda_t})^{1-\psi}$, a term that reflects the relative thickness of the job-supply and job-demand tails. The parameter $\lambda_t$ governs the tail thinness of the labor-demand distribution---smaller values imply a thicker upper tail and thus a larger mass of high-type job offers, while $\lambda_x$ plays the analogous role for labor supply, capturing the prevalence of high-skill workers. The ratio $\frac{\lambda_x}{\lambda_t}$ thus measures the relative scarcity of high-skill labor compared with high-type job offers: a higher value means firms face fiercer competition for top talent, pushing the wage gradient steeper to attract skilled workers.

From equilibrium matching (\ref{pam}), we derive the sorting relationship that links worker type to firm type: 
\begin{align}
h=\mu_t(x) = \frac{\lambda_x}{\lambda_t}x.\label{eq: sorting}
\end{align}
Combining (\ref{eq: wagex}) with \eqref{eq: sorting}, equation (\ref{eq: wagex}) can be equivalently rewritten as
 \begin{align}\label{eq: wageh}
\ln(w_t(\mu_t^{-1}(h)))=\frac{\psi}{\gamma}(\frac{\lambda_t}{\lambda_x})^{\psi}h+\ln(w_t(0)),
\end{align}
which characterizes how wages vary across different firms.

It is informative to examine firms' problem given capital stock $k$ and wedge $\tau_1$. Substituting \eqref{eq: wageh} into \eqref{eq: lFOC} yields the optimal labor input:

\begin{align}\label{eq: lh}
    l^*_t(\theta, k, \chi, \tau_1)=[\frac{w_t(0)\tau_1}{\gamma\chi A_tk^\alpha}]^{\frac{1}{\gamma-1}}\exp (\frac{\gamma-\psi}{\gamma(1-\gamma)}(\frac{\lambda_t}{\lambda_x})^{\psi}\theta)
\end{align}
Note that in our model $k$ and $\tau_1$ depend on state variables $(\theta,\epsilon_1,\epsilon_2)$. However, (\ref{eq: lh}) applies also to a scenario where firm-level capital stock and distortions are prefixed. The difference $\gamma-\psi$ captures the trade-off between worker quality and quantity. If $\gamma<\psi$, the demand for quality overwhelms the demand for scale: The wage schedule (\ref{eq: wagex}) is so steep that firms must shrink their size to afford higher worker quality, holding $k$ and $\tau_1$ fixed. In this case, high-type firms optimally operate with a small team of elite workers. This stands in contrast to a standard framework without worker heterogeneity (a special case of our model with $\psi=0$), where higher-type firms always hire more workers.

With (\ref{eq: wageh}) and firms' optimality conditions, we can readily express our variables of interests as functions of firm-level state variables.
\begin{lem}\label{lem: firmsolution}
   Suppose that in period $t$ jobs $h$ are exponentially distributed with parameter $\lambda_t$, that is $h \sim \text{Exp}\left(\lambda_t\right)$. Then the equilibrium output $Q$, rented capital $k$, marginal cost $\chi$ and labor hired $l$  at firm of type $(\theta, \epsilon_1, \epsilon_2)$ are as follows:
   \begin{align}
       Q_t(\theta, \epsilon_1, \epsilon_2)&=\bar{Q}_t\exp\left(\eta^Q \left(\eta_{\theta t}^Q\theta-\gamma \epsilon_1-\alpha\epsilon_2\right)\right), \label{eq: Qstar}\\
       k_t(\theta, \epsilon_1, \epsilon_2)&=\bar{k}_t \exp \left(\frac{\xi-1}{\xi} \eta^Q \left(\eta_{\theta t}^Q\theta-\gamma \epsilon_1-(\alpha+\frac{\xi}{\xi-1}\frac{1}{\eta^Q})\epsilon_2\right)  \right),\label{eq: kstar}\\ 
       \chi_t(\theta, \epsilon_1, \epsilon_2)&= \bar{\chi}_t \exp\left(-\frac{\eta^Q}{\xi} \left(\eta_{\theta t}^Q\theta-\gamma \epsilon_1-\alpha\epsilon_2\right)\right),\label{eq: chistar}\\ 
       l_t(\theta, \epsilon_1, \epsilon_2)&= \bar{l}_t\exp \left(\eta_{\theta t}^{l}\theta -\left(\frac{\xi-1}{\xi}\eta^Q\gamma+1\right) \epsilon_1-\frac{\xi-1}{\xi}\eta^Q\alpha \epsilon_2\right) ,  \label{eq: lhstar}
   \end{align}
   where

 \begin{IEEEeqnarray*}{rclrclrcl}
    \eta^Q&\equiv&\frac{\xi}{1+(1-\alpha-\gamma)(\xi-1)}, \quad&\eta_{\theta t}^Q &\equiv& -\gamma z_t+(1-\psi)(\frac{\lambda_t}{\lambda_x})^{\psi}, \quad&
    \eta_{\theta t}^{l}& \equiv&\frac{\xi-1}{\xi}\eta^Q\eta_{\theta t}^Q -z_t-\frac{\psi}{\gamma}(\frac{\lambda_t}{\lambda_x})^{\psi}. 
 \end{IEEEeqnarray*}
   
The time-dependent but type-independent constants $\bar{Q}_t, \bar{k}_t, \bar{\chi}_t,\bar{l}_t$ are defined in the proof of Lemma \ref{lem: firmsolution} (Appendix \ref{app:firmsolution}).
\end{lem}

In equilibrium, if $\epsilon_1 = \epsilon_2 = 0$, the firm-level variables above, such as firm size $l$, follow a Pareto distribution, since the underlying firm type $\theta$ is exponentially distributed. When $\epsilon_1$ and $\epsilon_2$ are drawn from the assumed zero-mean normal distributions with nonzero variance, these variables follow a convolution of a Pareto distribution and a log-normal distribution. Importantly, because the log-normal component decays faster in the upper tail than the Pareto component, the resulting distributions retain Pareto tails, consistent with the fat-tailed firm size distributions observed in the data.

\deferred[lem: firmsolution]{\subsection{Proof of Lemma \ref{lem: firmsolution}\label{app:firmsolution}}
First, given $\bar{Q}_t$, we have
\begin{IEEEeqnarray*}{rClrCl}
    \bar{k}_t = \alpha\frac{\xi-1}{\xi}\frac{ \bar{Q}_t^{\frac{\xi-1}{\xi}}  Y_t^\frac{1}{\xi} }{R_t}, \qquad \bar{\chi}_t=\frac{\xi-1}{\xi}\bar{Q}_t^{-\frac{1}{\xi}} Y_t^\frac{1}{\xi} , \qquad  \bar{l}_t= \left[\bar{\chi_t}A_t\bar{k}_t^{\alpha}\gamma/w_t(0) \right]^\frac{1}{1-\gamma}.
\end{IEEEeqnarray*}
Using the identity $$\bar{Q}_t=A_t\bar{k}_t^\alpha\bar{l}_t,$$ we obtain
\[
\bar{Q}_t = \left[A_t \left(\frac{\gamma}{w_t(0)}\right)^\gamma\left(\frac{\alpha}{R_t}\right)^\alpha\left(\frac{\xi-1}{\xi}Y_t^\frac{1}{\xi}\right)^{\alpha+\gamma} \right]^{\eta^Q} .
\]

The markup for an intermediate good is a constant $\frac{\xi}{\xi-1}$. From \eqref{demand},
\begin{equation}\label{eq: chi}
    \chi_{jt}=\left(\frac{Q_{jt}}{Y_t}\right)^{-\frac{1}{\xi}}\frac{\xi-1}{\xi}
\end{equation}
The first-order condition with respect to capital yields:
\begin{equation}\label{eq: kFOC}
    k_{jt}=\frac{\alpha \left(\frac{Q_{jt}}{Y_t}\right)^{-\frac{1}{\xi}}\frac{\xi-1}{\xi} Q_{jt}}{\tau_{2jt} R_t}.
\end{equation}
Further, substituting \eqref{eq: lh} and \eqref{eq: wageh} into \eqref{eq: lFOC} gives
\begin{align}
    Q_{jt}=\chi_{jt}^{-1}\tau_{1jt}w_t(0)\exp\left(\frac{1-\psi}{1-\gamma}(\frac{\lambda_t}{\lambda_x})^{\psi}\theta_{jt}\right)\left[\left(\frac{\gamma A_t\chi_{jt}}{w_t(0)\tau_{1jt} )}\right)  k_{jt}^{\alpha} \right]^{\frac{1}{1-\gamma}}.
\end{align}
Next, substituting for $\chi$ from \eqref{eq: chi} and for $k$ from \eqref{eq: kFOC} results in
\begin{align*}
  Q_{jt}&=w_t(0)\tau_{1jt}\exp\left(\frac{1-\psi}{1-\gamma}(\frac{\lambda_t}{\lambda_x})^{\psi}\theta_{jt}\right)\\
  &\quad\times\left[\left(\frac{Q_{jt}}{Y_t}\right)^{-\frac{1}{\xi}}\frac{\xi-1}{\xi}\right]^{-1}\left[\left(\frac{\gamma A_t\left(\frac{Q_{jt}}{Y_t}\right)^{-\frac{1}{\xi}}\frac{\xi-1}{\xi}}{w_t(0)\tau_{1jt} )}\right) \left(\frac{\alpha \left(\frac{Q_{jt}}{Y_t}\right)^{-\frac{1}{\xi}}\frac{\xi-1}{\xi} Q_{jt}}{\tau_{2jt} R_t}\right)^{\alpha} \right]^{\frac{1}{1-\gamma}},  
\end{align*}
where
\[
 \tau_1=\exp(z_t\theta+\epsilon_1),
\]
\[
 \tau_2=\exp(\epsilon_2),
\]
and after simple rearrangements, \eqref{eq: Qstar} follows. Similarly, \eqref{eq: chistar} and \eqref{eq: kstar} follow from substituting \eqref{eq: Qstar} into \eqref{eq: chi} and \eqref{eq: kFOC}, respectively. Finally, \eqref{eq: lhstar} is the result of substituting \eqref{wedge}, \eqref{eq: kstar} and \eqref{eq: chistar} into \eqref{eq: lh}. 
}

\subsection{Equilibrium Job Distribution}\label{sec:3.3}
Next we proceed to 
characterize $\lambda_t$.
The PDF of $h$ follows
\begin{align}
   f_t(h)&= \int_{-\infty}^{\infty}\int_{-\infty}^{\infty} l_t(h,\epsilon_1,\epsilon_2)*\lambda_\theta\exp(-\lambda_\theta h)\mathrm{d} \Phi(\frac{\epsilon_1}{\sigma_1}) \mathrm{d} \Phi(\frac{\epsilon_2}{\sigma_2})\nonumber\\
   &=\lambda_\theta\overline{l}_t\int_{-\infty}^{\infty}\exp\left(-\left(\frac{\xi-1}{\xi}\eta^Q\gamma+1\right) \epsilon_1\right)\mathrm{d} \Phi(\frac{\epsilon_1}{\sigma_1})\int_{-\infty}^{\infty}\exp\left(-\frac{\xi-1}{\xi}\eta^Q\alpha \epsilon_2\right)\mathrm{d} \Phi(\frac{\epsilon_2}{\sigma_2})\nonumber\\
   &\times\exp\left(\eta_{\theta t}^{l}h-\lambda_\theta h\right)\nonumber\\
   &=\tilde{l}_t*\exp\left(\eta_{\theta t}^{l}h-\lambda_\theta h\right),
\end{align}
where $\tilde{l}_t$ is a constant independent to $h$. 

Therefore $\lambda_t$ solves the following equation
\begin{align}
    -\lambda_t&=\eta_{\theta t}^{l}-\lambda_\theta\nonumber\\
    &=\frac{\big((\xi-1)\big(\gamma-\psi(1-\alpha)\big)-\psi\big)}
{\gamma\big(1+(1-\alpha-\gamma)(\xi-1)\big)}
\left(\frac{\lambda_t}{\lambda_x}\right)^{\psi}
-\frac{1+(\xi-1)(1-\alpha)}{1+(1-\alpha-\gamma)(\xi-1)}\,z_t
-\lambda_\theta.\label{eqm:jobdist}
\end{align}
\begin{lem}
    There exists a unique solution $\lambda_t\in [0,\infty)$ to equation (\ref{eqm:jobdist}).\label{lem:jobdist}
\end{lem}
\begin{proof}
    See Appendix \ref{app:jobdist}.
\end{proof}

\deferred[lem:jobdist]{\subsection{Proof of Lemma \ref{lem:jobdist}\label{app:jobdist}}
The LHS of equation (\ref{eqm:jobdist}) is decreasing $\lambda_t$ and ranges from $0$ to $-\infty$. We denote the RHS as $g(\lambda_t)$. Our objective is to prove that the equation $G(\lambda_t)\equiv g(\lambda_t)+\lambda_t=0$ has a unique solution $\lambda_t\in [0,\infty)$. 

First, let us consider the scenario where
\[
(\xi-1)\big(\gamma-\psi(1-\alpha)\big)-\psi>0,
\]
i.e. $g'(x)\geq0$.
It is immediate $G(0)<0$, $\lim_{x\to\infty}G(x)\to\infty$, and $G'(x)>0$ for $x\in [0,\infty)$. It is straightforward $G(x)\equiv g(x)+x=0$ has a unique solution $x\in [0,\infty)$.

Second, let us consider the case
\[
(\xi-1)\big(\gamma-\psi(1-\alpha)\big)-\psi<0.
\]
Now we have $G''(x)>0$ for $x\in [0,\infty)$. Again, we have $G(0)<0$, $\lim_{x\rightarrow\infty}G(x)\rightarrow\infty$. Since \(G\) is continuous on \([0,\infty)\), the conditions \(G(0)<0\) and 
\(\lim_{x\to\infty}G(x)>0\) imply, by the Intermediate Value Theorem, the existence 
of a point \(x^\ast>0\) such that \(G(x^\ast)=0\). Hence a solution exists.

Because \(G''(x)>0\) for all \(x\ge 0\), the function \(f\) is strictly convex on 
\([0,\infty)\). Strict convexity implies that there exists a unique point 
\(m\in\mathbb{R}\) at which \(G\) attains its global minimum, and that  
\(G'(m)=0\). We distinguish two cases according to the location of \(m\).

\smallskip
\noindent\emph{Case 1: \(m \le 0\).}
For any \(x>m\), strict convexity implies \(G'(x) > G'(m) = 0\). In particular,
\(G'(x) > 0\) for all \(x \ge 0\), and therefore \(G\) is strictly increasing on
\([0,\infty)\). Since \(G(0) < 0\) and \(G\) is strictly increasing, the equation 
\(G(x)=0\) has at most one solution in \([0,\infty)\). Combined with existence, this
solution is unique.

\smallskip
\noindent\emph{Case 2: \(m > 0\).}
Strict convexity implies that \(G\) is strictly decreasing on \([0,m]\) and strictly
increasing on \([m,\infty)\). Since \(m\) is the unique minimizer and 
\(G(0)<0\), we have \(G(m)\le G(0)<0\). Consequently,
\begin{equation}
    G(x) < 0 \qquad \text{for all } x \in [0,m],
\end{equation}
so no solution to \(G(x)=0\) can lie in this interval. Any solution must lie in 
\((m,\infty)\), where \(G\) is strictly increasing. Thus at most one such solution 
exists. Together with existence, the solution is unique.

In both cases, the equation \(G(x)=0\) has exactly one solution on 
\([0,\infty)\), which completes the proof.
}

The existence of this solution confirms our conjecture that an equilibrium with an exponential job offer distribution, $h \sim \text{Exp}(\lambda_t)$, indeed exists.  Our model falls into the generic class studied by \cite{eeckhout2018}, who establish the uniqueness of the equilibrium matching function and firm size distribution in this environment. Taken together with our result, their uniqueness theorem implies that the exponential distribution we conjecture and verify is the only distribution consistent with equilibrium in our setting.

Figure \ref{fig:jobdist} graphically depicts the two sides of equation (\ref{eqm:jobdist}) as functions of $\lambda_t$, represented by the two curves. Their intersection identifies a unique equilibrium value $\lambda_t > 0$. Combined with the parameter governing the distribution of worker skills, $\lambda_x$, this equilibrium allows us to pin down the assortative matching pattern between firms and workers described in equation (\ref{pam}).

\begin{figure}[htbp]
    \centering
    \includegraphics[width=0.8\textwidth]{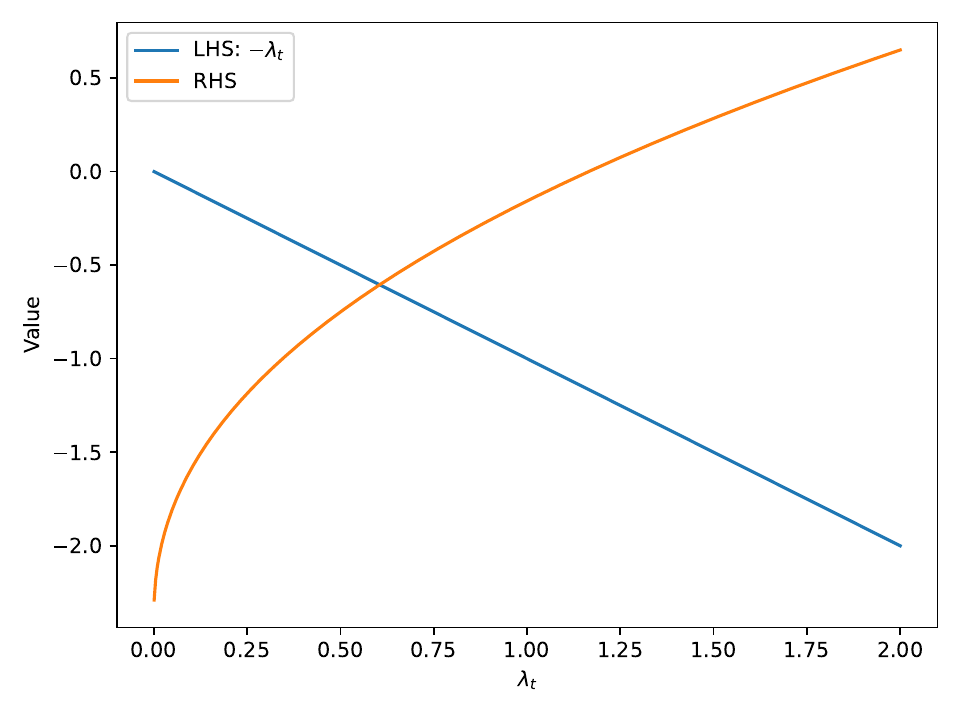}
    \caption{Equilibrium determination of $\lambda_t$.}
    \label{fig:jobdist}
\end{figure}

Note that we must still impose $\int_{0}^{\infty}f_t(h)\mathrm{d}h=1$, or equivalently, $\tilde{l}_t=\lambda_t$. Together with equation (\ref{eqm:jobdist}), this identity completes the labor-market clearing condition. Recall that
\begin{align*}
   \tilde{l}_t&\equiv \lambda_\theta\overline{l}_t\int_{-\infty}^{\infty}\exp\left(-\left(\frac{\xi-1}{\xi}\eta^Q\gamma+1\right) \epsilon_1\right)\mathrm{d} \Phi(\frac{\epsilon_1}{\sigma_1})\int_{-\infty}^{\infty}\exp\left(-\frac{\xi-1}{\xi}\eta^Q\alpha \epsilon_2\right)\mathrm{d} \Phi(\frac{\epsilon_2}{\sigma_2})\\
   &=\lambda_\theta\overline{l}_t\exp\!\Big(\frac{1}{2}\big[\big(\frac{\xi-1}{\xi}\eta^Q\gamma+1\big)^2\sigma_1^{2}+\big(\frac{\xi-1}{\xi}\eta^Q\alpha\big)^2\sigma_2^{2}\big]\Big).
\end{align*}
where $\overline{l}_t$, derived in Appendix \ref{app:firmsolution}, is given by
\[
\overline{l}_t=\left[\left(\frac{\gamma A_t \left(\frac{\bar{Q}_t}{Y_t}\right)^{-\frac{1}{\xi}}\frac{\xi-1}{\xi}}{w_t(0) )}\right) \left(\frac{\alpha \left(\frac{\bar{Q}_t}{Y_t}\right)^{-\frac{1}{\xi}}\frac{\xi-1}{\xi} \bar{Q}_{t}}{R_t}\right)^{\alpha} \right]^{\frac{1}{1-\gamma}}.
\]
$\overline{l}_t$ is a function of aggregate variables $w_t(0)$, $R_t$, and $Y_t$. We next move to characterize these objects.

\subsection{Aggregate Variables}
\paragraph{Prices and aggregate output} We can use market clearing conditions to characterize the remaining equilibrium variables. 
First, as noted above,
the labor market clearing condition $\int_{0}^{\infty}f_t(h)\mathrm{d}h=1$, implies that
\begin{align}
\lambda_\theta B_1\left[\left(\frac{\gamma A_t \left(\frac{\bar{Q}_t}{Y_t}\right)^{-\frac{1}{\xi}}\frac{\xi-1}{\xi}}{w_t(0) )}\right) \left(\frac{\alpha \left(\frac{\bar{Q}_t}{Y_t}\right)^{-\frac{1}{\xi}}\frac{\xi-1}{\xi} \bar{Q}_{t}}{R_t}\right)^{\alpha} \right]^{\frac{1}{1-\gamma}}=\lambda_t, \label{laborclear}
\end{align}
where
\[
B_1\equiv \exp\!\Big(\frac{1}{2}\big[\big(\frac{\xi-1}{\xi}\eta^Q\gamma+1\big)^2\sigma_1^{2}+\big(\frac{\xi-1}{\xi}\eta^Q\alpha\big)^2\sigma_2^{2}\big]\Big).
\]
Equation \eqref{laborclear} provides a relationship among $w_t(0)$, $R_t$, and $Y_t$. Importantly, $\lambda_t$ is separately characterized from the job distribution equation (\ref{eqm:jobdist}), which does not involve $w_t(0)$, $R_t$, and $Y_t$.

Next, capital market clearing requires aggregate capital supply to equal aggregate capital demand:
\begin{align*}
   K_t= \int_{-\infty}^{\infty}\int_{-\infty}^{\infty} \int_{0}^{\infty}k_t(\theta,\epsilon_1,\epsilon_2)*\lambda_\theta\exp(-\lambda_\theta \theta)\mathrm{d}\theta\mathrm{d} \Phi(\frac{\epsilon_1}{\sigma_1}) \mathrm{d} \Phi(\frac{\epsilon_2}{\sigma_2}),
\end{align*}
where the LHS, $K_t$, represents capital supply, an aggregate state variable predetermined at $t$, while the RHS aggregates firms' capital demand. Substituting equation \eqref{eq: kstar} and the expression of $\bar{k}_t$ yields
\begin{align}
   K_t&=\lambda_\theta\bar{k}_t\times\int_{0}^{\infty}\exp \left(\frac{\xi-1}{\xi} \eta^Q \eta_{\theta t}^Q\theta-\lambda_\theta \theta\right)\mathrm{d}\theta\nonumber \\
   &\times\int_{-\infty}^{\infty}\exp\left(-\frac{\xi-1}{\xi} \eta^Q \gamma\epsilon_1\right)\mathrm{d} \Phi(\frac{\epsilon_1}{\sigma_1})\times\int_{-\infty}^{\infty}\exp\left(-\frac{\xi-1}{\xi} \eta^Q (\alpha+\frac{\xi}{\xi-1}\frac{1}{\eta^Q})\epsilon_2\right)\mathrm{d} \Phi(\frac{\epsilon_2}{\sigma_2})  \nonumber\\
  &=\lambda_\theta B_2 \alpha\frac{\xi-1}{\xi}\frac{ \bar{Q}_t^{\frac{\xi-1}{\xi}}  Y_t^\frac{1}{\xi} }{R_t}\label{capitalclear}
\end{align}
where
\[
B_2\equiv\frac{
\exp\!\left[
\frac{1}{2}
\left(
\left(\frac{\xi-1}{\xi}\eta^{Q}\gamma\right)^{2}\sigma_1^{2}
+
\left(\frac{\xi-1}{\xi}\alpha\eta^{Q}+1\right)^{2}\sigma_2^{2}
\right)
\right]
}{
\displaystyle
\lambda_\theta
-
\frac{\xi-1}{\xi}\eta^{Q}\eta_{\theta t}^{Q}
}
\]
provided that $\lambda_\theta
-
\frac{\xi-1}{\xi}\eta^{Q}\eta_{\theta t}^{Q}>0$, otherwise capital demand becomes unbounded. Equation \eqref{capitalclear} therefore delivers a second equilibrium condition linking $w_t(0)$, $R_t$, and $Y_t$. Combining it with equation \eqref{laborclear}, we can represent prices $w_t(0)$ (and the entire wage schedule) and $R_t$ as functions of aggregate state variables and aggregate output $Y_t$.

Finally, the aggregate output $Y_t$ follows 
\begin{align}
     Y_t=\int_{-\infty}^{\infty}\int_{-\infty}^{\infty}\int_{0}^{\infty} \frac{\xi}{\xi-1}\chi_{t}^*(\theta,\epsilon_1,\epsilon_2) Q_{t}^*(\theta,\epsilon_1,\epsilon_2)\lambda_\theta\exp(-\lambda_\theta \theta)\mathrm{d}\theta\mathrm{d} \Phi(\frac{\epsilon_1}{\sigma_1}) \mathrm{d} \Phi(\frac{\epsilon_2}{\sigma_2}).\label{goodsclear}
\end{align}
The supply of each intermediate good is equal to its demand
\begin{align}
     Q_{t}^*(\theta,\epsilon_1,\epsilon_2)=(\frac{\xi}{\xi-1}\chi_{t}^*(\theta,\epsilon_1,\epsilon_2))^{-\xi}Y_t, 
\end{align}
which be substituted into \eqref{goodsclear} to get
\begin{align}
     1=\int_{-\infty}^{\infty}\int_{-\infty}^{\infty}\int_{0}^{\infty} \left[\frac{\xi}{\xi-1}\chi_{t}^*(\theta,\epsilon_1,\epsilon_2) \right]^{1-\xi}\lambda_\theta\exp(-\lambda_\theta \theta)\mathrm{d}\theta\mathrm{d} \Phi(\frac{\epsilon_1}{\sigma_1}) \mathrm{d} \Phi(\frac{\epsilon_2}{\sigma_2}).\label{goodsclear2}
\end{align}
Hence
\begin{align}
 1&=   \lambda_\theta\left( \frac{\xi}{\xi - 1} \right)^{1 - \xi}\overline{\chi}_t^{1 - \xi}\times\int_{0}^{\infty}\exp \left(\frac{\xi-1}{\xi} \eta^Q \eta_{\theta t}^Q\theta-\lambda_\theta \theta\right)\mathrm{d}\theta\nonumber \\
   &\times\int_{-\infty}^{\infty}\exp\left(-\frac{\xi-1}{\xi} \eta^Q \gamma\epsilon_1\right)\mathrm{d} \Phi(\frac{\epsilon_1}{\sigma_1})\times\int_{-\infty}^{\infty}\exp\left(-\frac{\xi-1}{\xi} \eta^Q \alpha\epsilon_2\right)\mathrm{d} \Phi(\frac{\epsilon_2}{\sigma_2})  \nonumber\\
  &=\lambda_\theta B_3 \bar{Q}_t^{\frac{\xi-1}{\xi}} Y_t^\frac{1-\xi}{\xi}, \label{goodsclear3}
\end{align}
where
\[
B_3\equiv\frac{
\exp\!\left[
\frac{1}{2}
\left(
\left(\frac{\xi-1}{\xi}\eta^{Q}\gamma\right)^{2}\sigma_1^{2}
+
\left(\frac{\xi-1}{\xi}\alpha\eta^{Q}\right)^{2}\sigma_2^{2}
\right)
\right]
}{
\displaystyle
\lambda_\theta
-
\frac{\xi-1}{\xi}\eta^{Q}\eta_{\theta t}^{Q}
}.
\]
Equation \eqref{goodsclear3} is therefore another equilibrium condition of $w_t(0)$, $R_t$,  and $Y_t$. Together with \eqref{laborclear} and \eqref{capitalclear}, the three equations determine the three unknown variables $\left\{w_t(0), R_t, Y_t\right\}$ uniquely for given aggregate state variables. The results are presented in Appendix \ref{app:unknownvariables}.

\deferred[unknown]{\subsection{Derivation of $\left\{w_t(0), R_t, Y_t\right\}$\label{app:unknownvariables}}
For notation ease, define
\[
\kappa \equiv \frac{\xi-1}{\xi}.
\]

We have
\[
\lambda_\theta B_3 \bar{Q}_t^{\frac{\xi-1}{\xi}} Y_t^{\frac{1-\xi}{\xi}} = 1.
\]
Using $\kappa=\frac{\xi-1}{\xi}$:
\[
\lambda_\theta B_3 \bar Q_t^{\kappa} Y_t^{-\kappa} = 1.
\]
Thus,
\[
\left(\frac{Y_t}{\bar Q_t}\right)^\kappa = \lambda_\theta B_3
\quad\Rightarrow\quad
Y_t = \bar Q_t (\lambda_\theta B_3)^{1/\kappa}
     = \bar Q_t(\lambda_\theta B_3)^{\frac{\xi}{\xi-1}}.
\]

Define
\[
M \equiv (\lambda_\theta B_3)^{1/\kappa}
    ,
\]
so
\[
Y_t = M\bar Q_t.
\]

Next, we have
\[
\lambda_\theta B_2 \alpha\kappa
\frac{\bar Q_t^\kappa Y_t^{1/\xi}}{R_t}
= K_t.
\]

Substitute $Y_t=M\bar Q_t$:
\[
\bar Q_t^{\kappa}Y_t^{1/\xi}
= \bar Q_t^\kappa(M\bar Q_t)^{1/\xi}
= M^{1/\xi}\bar Q_t^{\kappa+1/\xi}.
\]

Thus,
\[
R_t
= \frac{\lambda_\theta B_2 \alpha\kappa\, M^{1/\xi}\bar Q_t^{\kappa+1/\xi}}{K_t}.
\]

Since $\kappa+\frac{1}{\xi}=1$:
\[
R_t
= \frac{\lambda_\theta B_2 \alpha\kappa\, M^{1/\xi}\bar Q_t}{K_t}.
\]

Using $M^{1/\xi}=(\lambda_\theta B_3)^{\frac{1}{\xi-1}}$:
\[
R_t
= \frac{
\lambda_\theta B_2 \alpha\kappa\;
\bar Q_t
(\lambda_\theta B_3)^{\frac{1}{\xi-1}}
}{K_t}.
\]

The labor market clearing equation is
\[
\lambda_\theta B_1\left[
\left(
\frac{
\gamma A_t (\frac{\bar Q_t}{Y_t})^{-1/\xi}\kappa
}{w_t(0)}
\right)
\left(
\frac{
\alpha (\frac{\bar Q_t}{Y_t})^{-1/\xi}\kappa\bar Q_t
}{R_t}
\right)^\alpha
\right]^{\frac{1}{1-\gamma}}
= \lambda_t.
\]

Since
\[
\left(\frac{\bar Q_t}{Y_t}\right)^{-1/\xi}
= M^{1/\xi},
\]
we get
\[
\lambda_\theta B_1
\left[
\frac{\gamma A_t M^{1/\xi}\kappa}{w_t(0)}
\left(
\frac{\alpha M^{1/\xi}\kappa\bar Q_t}{R_t}
\right)^\alpha
\right]^{\frac{1}{1-\gamma}}
= \lambda_t.
\]

Now use
\[
\frac{\alpha M^{1/\xi}\kappa \bar Q_t}{R_t}
= \frac{K_t}{\lambda_\theta B_2}.
\]

Thus,
\[
\lambda_\theta B_1
\left[
\frac{\gamma A_t M^{1/\xi}\kappa}{w_t(0)}
\left(\frac{K_t}{\lambda_\theta B_2}\right)^\alpha
\right]^{\frac{1}{1-\gamma}}
= \lambda_t.
\]

Raise both sides to power $(1-\gamma)$:
\[
\frac{\gamma A_t M^{1/\xi}\kappa}{w_t(0)}
\left(\frac{K_t}{\lambda_\theta B_2}\right)^\alpha
=
\left(\frac{\lambda_t}{\lambda_\theta B_1}\right)^{1-\gamma}.
\]

Solve for $w_t(0)$:
\[
w_t(0)
= \gamma A_t M^{1/\xi}\kappa
  \left(\frac{K_t}{\lambda_\theta B_2}\right)^\alpha
  \left(\frac{\lambda_t}{\lambda_\theta B_1}\right)^{\gamma-1}.
\]

Since $M^{1/\xi}=(\lambda_\theta B_3)^{\frac{1}{\xi-1}}$:
\[
w_t(0)
=
A_t\gamma\kappa(\lambda_\theta B_3)^{\frac{1}{\xi-1}}
\left(\frac{\lambda_t}{\lambda_\theta B_1}\right)^{\gamma-1}
\left(\frac{K_t}{\lambda_\theta B_2}\right)^\alpha.
\]

We are now ready to solve $\bar Q_t$.
We have:
\[
\bar Q_t
= \left[
A_t
\left(\frac{\gamma}{w_t(0)}\right)^{\gamma}
\left(\frac{\alpha}{R_t}\right)^{\alpha}
\left(\kappa Y_t^{1/\xi}\right)^{\alpha+\gamma}
\right]^{\eta^{Q}}.
\]

Use $Y_t=M\bar Q_t$:
\[
Y_t^{1/\xi}
= (M\bar Q_t)^{1/\xi}
= M^{1/\xi}\bar Q_t^{1/\xi}.
\]

Thus
\[
(\kappa Y_t^{1/\xi})^{\alpha+\gamma}
= \kappa^{\alpha+\gamma}
  M^{\frac{\alpha+\gamma}{\xi}}
  \bar Q_t^{\frac{\alpha+\gamma}{\xi}}.
\]

Also
\[
\left(\frac{\alpha}{R_t}\right)^\alpha
=
\left(
\frac{K_t}{\lambda_\theta B_2 \kappa}
\right)^\alpha
\bar Q_t^{-\alpha}
M^{-\alpha/\xi}.
\]

Combine terms:
\[
\bar Q_t
=
\big[
A_t K_t^\alpha
\kappa^\gamma
\lambda_\theta^{-\alpha}B_2^{-\alpha}
M^{\gamma/\xi}
(\gamma/w_t(0))^\gamma
\bar Q_t^{-\alpha + \frac{\alpha+\gamma}{\xi}}
\big]^{\eta^Q},
\]
and define
\[
C_{\mathrm{in}}
\equiv
A_t K_t^\alpha
\kappa^\gamma
\lambda_\theta^{-\alpha}B_2^{-\alpha}
M^{\gamma/\xi}
\left(\frac{\gamma}{w_t(0)}\right)^\gamma.
\]

Then the fixed point is
\[
\bar Q_t
= \left[
C_{\mathrm{in}}
\bar Q_t^{-\alpha + \frac{\alpha+\gamma}{\xi}}
\right]^{\eta^{Q}}.
\]
Take logs to get:
\[
\log\bar Q_t
= \eta^{Q}\log C_{\mathrm{in}}
+ \eta^Q\!\left(-\alpha + \frac{\alpha+\gamma}{\xi}\right)\!\log\bar Q_t.
\]
Then rearrange terms:
\[
\left[
1 - \eta^{Q}\left(-\alpha + \frac{\alpha+\gamma}{\xi}\right)
\right]
\log\bar Q_t
= \eta^{Q}\log C_{\mathrm{in}}.
\]
Thus
\[
\log\bar Q_t
=
\frac{\eta^{Q}}{
1 - \eta^{Q}\left(-\alpha + \frac{\alpha+\gamma}{\xi}\right)
}
\log C_{\mathrm{in}}.
\]
Finally, exponentiate:
\[
\bar Q_t
= C_{\mathrm{in}}^{\frac{
\eta^{Q}
}{
1 - \eta^{Q}\left(-\alpha + \frac{\alpha+\gamma}{\xi}\right)
}}.
\]

The final solutions can be represented as follows

\[
\begin{aligned}
w_t(0) &=
A_t\gamma\kappa(\lambda_\theta B_3)^{\frac{1}{\xi-1}}
\left(\dfrac{\lambda_t}{\lambda_\theta B_1}\right)^{\gamma-1}
\left(\dfrac{K_t}{\lambda_\theta B_2}\right)^\alpha,\\
\bar Q_t &= C_{\mathrm{in}}^{\frac{
\eta^{Q}
}{
1 - \eta^{Q}\left(-\alpha + \frac{\alpha+\gamma}{\xi}\right)
}},\\[6pt]
Y_t &= M\bar Q_t,\\[6pt]
R_t &= 
\dfrac{\lambda_\theta B_2\alpha\kappa\;\bar Q_t\;(\lambda_\theta B_3)^{\frac{1}{\xi-1}}}{K_t}.\\[6pt]
\end{aligned}
\]
}

\paragraph{Consumption and investment}Finally, we aggregate our previous results to characterize the dynamics of consumption, and ultimately the evolution of the endogenous aggregate state variable $K_t$. At the aggregate level, the model behaves analogously to a stochastic neoclassical growth model. Aggregate consumption satisfies the representative household's Euler equation
\begin{align}
   \beta E_t\frac{C_t}{C_{t+1}}[R_{t+1}+(1-\delta)]=1, 
   \end{align}
together with the intertemporal budget constraint
\begin{align}
      C_t+K_{t+1}&= (1-\delta)K_t+Y^l_t+Y^k_t+Y_t^d.
\end{align}
where $Y^l_t$, $Y_t^k$, and $Y^d_t$ respectively denote aggregate labor income, capital income, and distributed profits. These objects are given by
\begin{align}
    Y_t^l&=\frac{\gamma\lambda_\theta B_1\bar{\chi}_t\bar{Q}_t}{\lambda_\theta
-
\frac{\xi-1}{\xi}\eta^{Q}\eta_{\theta t}^{Q}+z_t},\\
Y_t^k&=\alpha\lambda_\theta B_2 \bar{\chi}_t\bar{Q}_t,\\
Y_t^d&=(\frac{\xi}{\xi-1}-\gamma-\alpha)\lambda_\theta B_3\bar{\chi}_t\bar{Q}_t.
\end{align}
The terms $\bar{\chi}_t$ and $\bar{Q}_t$ are obtained from our previous analysis. Note that $Y^l_t$, $Y_t^k$, and $Y^d_t$ do not add up to the aggregate output $Y_t$, because wedges from market distortions generate resource losses that drive a gap between factor payments and aggregate output.

\section{Micro Implications}\label{sec:micro}
The framework developed above enables an analytical examination of cyclical fluctuations in micro-level distributions. In this section, we focus on two key margins: wages, which serve as a quantitative measure of employee performance, and productivity, which reflects employer effectiveness. Specifically, we study the dispersion of these variables and analyze how it evolves over the business cycle.

We study business cycles driven by fluctuations in $z_t$. Higher-type firms are more sensitive to these changes, since the shocks have a larger impact on firms with greater $\theta$. These changes induce shifts in the composition of labor demand, which in turn alter the sorting pattern between firms and workers.

To begin with, consider how a decline in $z_t$ affects the equilibrium job distribution parameter $\lambda_t$. Recall that $\lambda_t$ is determined by equation 
\begin{align}
    -\lambda_t&=g(\lambda_t,z_t)\nonumber\\
    &\equiv\frac{\big((\xi-1)\big(\gamma-\psi(1-\alpha)\big)-\psi\big)}
{\gamma\big(1+(1-\alpha-\gamma)(\xi-1)\big)}
\left(\frac{\lambda_t}{\lambda_x}\right)^{\psi}
-\frac{1+(\xi-1)(1-\alpha)}{1+(1-\alpha-\gamma)(\xi-1)}\,z_t
-\lambda_\theta.\label{eqm:jobdist2}
\end{align}
Figure \ref{fig:lambda_comparative} illustrates the comparative statics following a decline in $z_t$: the function $g(\lambda_t,z_t)$ shifts upward, leading to a lower new equilibrium $\lambda_t$. A formal proof is provided in Appendix \ref{app:lambda}. The main result is summarized as follows.
\begin{lem}
When $z_t$ decreases, the equilibrium job distribution parameter $\lambda_t$ decreases.\label{lem:lambda}
\end{lem}

\deferred[lem:lambda]{\subsection{Proof of Lemma \ref{lem:lambda}\label{app:lambda}}
Consider the equation
\[
G(x; z) \equiv b\, x^{\psi} + x - z = 0, \qquad x \in [0,\infty),
\]
with $0<\psi<1$ and $z>0$. Define
\[
H(x) \equiv b\, x^\psi + x,
\]
so that the solution $x(z)$ satisfies $H(x) = z$. $G(x;z)$ corresponds to equation (\ref{eqm:jobdist2}) if we denote
\[
b\equiv\frac{\big((\xi-1)\big(\gamma-\psi(1-\alpha)\big)-\psi\big)}
{\gamma\big(1+(1-\alpha-\gamma)(\xi-1)\big)}
\left(\frac{1}{\lambda_x}\right)^{\psi},
\]
\[
z\equiv\frac{1+(\xi-1)(1-\alpha)}{1+(1-\alpha-\gamma)(\xi-1)}\,z_t
+\lambda_\theta.
\]
Our objective is to prove the solution of $G(x;z)=0$, $x(z)$ increases with $z$.

\textbf{Case 1: $b\ge 0$.} 
Then
\[
H'(x) = 1 + b \psi x^{\psi-1} > 0 \quad \text{for all } x>0,
\]
so $H$ is strictly increasing. By the implicit function theorem,
\[
\frac{dx}{dz} = \frac{1}{1 + b \psi x^{\psi-1}} > 0.
\]
Hence $x(z)$ is strictly increasing in $z$, and a decrease in $z$ necessarily decreases $x(z)$.

\textbf{Case 2: $b<0$.} 
Then
\[
H'(x) = 1 + b \psi x^{\psi-1}, \qquad H''(x) = b \psi (\psi-1) x^{\psi-2} > 0 \text{ for } x>0,
\]
so $H'$ is strictly increasing. There exists a unique $x^*>0$ such that $H'(x^*) = 0$. Therefore, $H$ is strictly decreasing on $(0,x^*)$ and strictly increasing on $(x^*,\infty)$. Since $H(0) = 0$ and $z>0$, the solution $x(z)$ must lie on the increasing branch, i.e.\ $x(z) > x^*$, where $H'(x(z)) > 0$. Applying the implicit function theorem again gives
\[
\frac{dx}{dz} = \frac{1}{1 + b \psi x^{\psi-1}} > 0.
\]

\textbf{Conclusion.} 
For all real $b$, under $0<\psi<1$ and $z>0$, the equation $b\, x^\psi + x - z = 0$ has a unique solution $x(z)\in[0,\infty)$, and $x(z)$ is strictly increasing in $z$. In particular, a decrease in $z$ necessarily decreases $x(z)$.
\hfill $\blacksquare$

}

\begin{figure}[ht]
    \centering
    \includegraphics[width=0.8\textwidth]{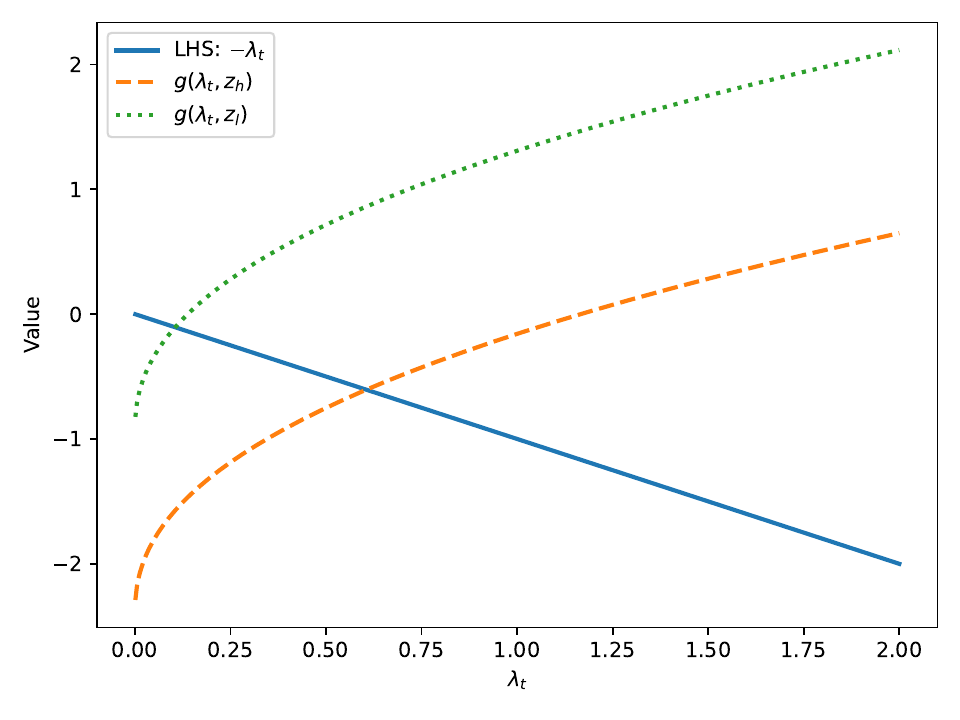}
    \caption{Equilibrium determination of $\lambda_t$ following an decrease of $z_t$}
     \caption*{\footnotesize\emph{Notes:} The figure plots a decrease of $z_t$ from $z_h$ to $z_l$. The RHS of equation (\ref{eqm:jobdist2}), $g(\lambda_t,z_t)$ shifts upward.}
    \label{fig:lambda_comparative}
\end{figure}

We can infer from the matching function (\ref{pam}) that during episodes of high $z_t$---which correspond to lower aggregate allocative efficiency---worker–firm sorting deteriorates, with workers being matched to inferior firms. This mechanism is reminiscent of the findings in \cite{crane2023},  who document that during downturns high-rank workers are more likely to work at low-rank firms, while low-rank workers are less likely to be employed by high-rank firms.

Note that $z_t$ is a reduced-form shock, designed to abstract from the micro-foundations of market frictions and their interactions with aggregate shocks. A decline in $z_t$ generates aggregate expansions in which better (or higher-type) firms expand relatively more. In reality, however, business cycles can arise from a variety of structural shocks, each potentially producing different patterns in micro-level distributions. In Section \ref{sec:4.3}, we use our theoretical framework to examine the effects of some of these structural drivers on labor market and productivity distributions.
\subsection{Wage Inequality} 
Recall that the wage schedule is given by
 \begin{align}
\ln(w_t(x))=\frac{\psi}{\gamma}(\frac{\lambda_x}{\lambda_t})^{1-\psi}x+\ln(w_t(0)).
\end{align}
Hence the dispersion of wages is
\begin{align}
\mathrm{Var}_t(\log w_{it})=(\frac{\psi}{\gamma})^2(\frac{\lambda_x}{\lambda_t})^{2-2\psi}\lambda_x^{-2}=(\frac{\psi}{\gamma})^2\lambda_x^{-2\psi}\lambda_t^{2\psi-2}.\label{eq:var_wage}
\end{align}
Since a decline in $z_t$ reduces $\lambda_t$, we obtain the following result.
\begin{prop}
   When $z_t$ decreases, wage inequality, $\mathrm{Var}_t(\log w_{it})$, increases. 
\end{prop}
Intuitively, a decline in $z_t$ enables high-type firms---those offering high-type jobs---to expand more aggressively. This shifts the composition of job offers toward the upper tail, increasing the relative mass of high-type offers in the labor market. Consequently, the wage schedule steepens as competition intensifies, even in the absence of any change in labor supply. This results in greater wage dispersion.

\subsection{Firm-level Productivity}
It is well established in macroeconomics that firms exhibit heterogeneous responses to aggregate shocks, particularly in terms of output and investment. What sets our framework apart is that this heterogeneity extends beyond output and investment to measured productivity, which also reacts differently across firms.

There are various ways to measure firms' productivity:
\begin{align}
    \log TFP_{jt}\coloneqq\log y_{jt}-\alpha\log k_{jt}-\gamma\log l_{jt}
\end{align}
where $y_{jt}$ can refer to either the output quantity or revenue of firm $j$, depending on whether we aim to measure quantity-based productivity (TFPQ), or revenue-based productivity (TFPR).

Let us first consider the physical productivity, i.e. TFPQ. Using (\ref{prod2}) and (\ref{pam}), we obtain
\begin{align}\label{eq: tfpq}
   \log TFPQ(\theta)=\left(\mu^{-1}_t(\theta)\right)^\psi\theta^{1-\psi}
   =(\frac{\lambda_t}{\lambda_x})^{\psi}\theta.
\end{align}
Thus, TFPQ reflects only the intrinsic productivity of the firm and does not depend on firm-level distortions $\epsilon_1$ or $\epsilon_2$, which capture external market frictions. Its dependence on the aggregate state $z_t$ is only indirect, operating through the effect on the job-type distribution parameter $\lambda_t$. Similar to the distribution of wages, the distribution of TFPQ across firms is also Pareto.  

The dispersion of $\log TFPQ_{jt}$ at time $t$ is therefore
$$\mathrm{Var}_t(\log TFPQ_{jt})=(\frac{\lambda_t}{\lambda_x})^{2\psi} \mathrm{Var}\left( \theta^{1-\psi}\right)=(\frac{\lambda_t}{\lambda_x})^{2\psi}\lambda_{\theta}^{-2}$$ 
We already know that a decline in $z_t$ reduces $\lambda_t$. This immediately yields the following result.
\begin{prop}\label{prop: tfpq}
   When $z_t$ decreases, the dispersion of firm-level $\log TFPQ$, $\mathrm{Var}_t(\log TFPQ_{jt})$, decreases. 
\end{prop}

Intuitively, when $z_t$ falls, high-type firms expand more aggressively than their low-type counterparts. This expansion shifts the composition of the labor market: high-type firms must increasingly hire lower-skill workers. As a result, their physical productivity is pulled down relative to that of low-type firms, compressing the overall dispersion of TFPQ.

Next, we turn to revenue-based productivity, TFPR. By definition, TFPR is the product of TFPQ and the price of the intermediate good. From (\ref{eq: chistar}), we obtain
\begin{align}
   \log TFPR_t\left(\theta,\epsilon_1,\epsilon_2\right)&=\log\frac{\xi}{\xi-1}-\frac{\eta^Q}{\xi} \left(\eta_{\theta t}^Q\theta-\gamma \epsilon_1-\alpha\epsilon_2\right)+ (\frac{\lambda_t}{\lambda_x})^{\psi}\theta\nonumber\\
   &=\log\frac{\xi}{\xi-1}+\left[(\frac{\lambda_t}{\lambda_x})^{\psi}-\frac{\eta^Q\eta^Q_{\theta t}}{\xi}\right]\theta+\frac{\eta^Q}{\xi} \left(\gamma \epsilon_1+\alpha\epsilon_2\right)
\end{align}
Unlike TFPQ, TFPR responds directly to firm-level distortions, as these distortions shift marginal costs and thus affect pricing decisions. Through its price component, TFPR therefore reflects not only "pure" production efficiency but also the impact of distortions and demand conditions. Importantly, higher TFPR may actually be associated with greater distortions, since distortions can raise marginal costs and thus inflate prices. As with firm sizes, TFPR also displays a Pareto tail.  

The dispersion of $\log\text{TFPR}_{jt}$ at time $t$ is given by
\begin{equation}
\mathrm{Var}_t(\log TFPR_{jt})=\left[(\frac{\lambda_t}{\lambda_x})^{\psi}-\frac{\eta^Q\eta^Q_{\theta t}}{\xi}\right]^2\lambda_{\theta}^{-2}+(\frac{\eta^Q}{\xi})^2\left(\gamma^2\sigma_1^2+\alpha^2\sigma_2^2\right)\label{var_tfpr}
\end{equation}
We can now state the main result.
\begin{prop}
\[
(\frac{\lambda_t}{\lambda_x})^{\psi}-\frac{\eta^Q\eta^Q_{\theta t}}{\xi}>0.
\]
 Moreover,  when $z_t$ decreases, the dispersion of firm-level revenue productivity, $\mathrm{Var}_t(\log TFPR_{jt})$, decreases. \label{prop:tfpr}
\end{prop}
\begin{proof}
 We have
\begin{align*}
(\frac{\lambda_t}{\lambda_x})^{\psi}-\frac{\eta^Q\eta^Q_{\theta t}}{\xi}&=(\frac{\lambda_t}{\lambda_x})^{\psi}+\frac{1}{1+(1-\alpha-\gamma)(\xi-1)}\big[\gamma z_t-(1-\psi)(\frac{\lambda_t}{\lambda_x})^{\psi}\big]\\
&=\big[1-\frac{1-\psi}{1+(1-\alpha-\gamma)(\xi-1)}\big](\frac{\lambda_t}{\lambda_x})^{\psi}+\frac{\gamma}{1+(1-\alpha-\gamma)(\xi-1)}z_t.
\end{align*}
Both terms on the RHS are strictly positive. Moreover, a reduction in $z_t$ reduces each term, and therefore reduces the expression as a whole. Q.E.D.  
\end{proof}

The change in TFPR dispersion reflects shifts in both the dispersion of TFPQ and the dispersion of prices (or, equivalently, marginal costs). 
Mathematically, the impact of $z_t$ operates through its effect on the term $$(\frac{\lambda_t}{\lambda_x})^{\psi}-\frac{\eta^Q\eta^Q_{\theta t}}{\xi},$$ The first component, $(\lambda_t/\lambda_x)^{\psi}$, always decreases when $z_t$ falls, as shown earlier. The second component,
\begin{equation}
-\frac{\eta^Q\eta^Q_{\theta t}}{\xi}=\frac{1}{1+(1-\alpha-\gamma)(\xi-1)}\big[\gamma z_t-(1-\psi)(\frac{\lambda_t}{\lambda_x})^{\psi}\big],\label{eq:mc}
\end{equation}
embeds two opposing forces. A lower $z_t$ reduces the dispersion of firm-level distortions, compressing the dispersion of marginal costs and therefore prices---this corresponds to the first term on the RHS.  At the same time, the change in the matching pattern increases marginal costs through $-(\lambda_t/\lambda_x)^{\psi}$, a force increasing the dispersion of marginal costs, as indicated by the second term on the RHS. However, as shown in the previous proof, the forces is outweighed by the decline in $(\lambda_t/\lambda_x)^{\psi}$, i.e. the direct effect operating through TFPQ. Consequently, the overall term unambiguously decreases when $z_t$ falls, implying that the dispersion of TFPR decreases as well.

\paragraph{Robustness}

As the response of the productivity distribution to changes in the matching functions is the main result of this article, it is natural to wonder about the extent to which it is driven by our functional form assumptions. In our specification improvements in firms' matches increase log TFPQ multiplicatively. While this very particular form of dependence is unlikely to hold in more general models, it is also clearly not necessary for the larger conclusion to hold. Indeed, what suffices is that improvements in firms' matches improve the distribution of TFPQ in a strong enough sense, specifically in the hazard rate order (HRO) sense \citep[see Section 1.B in][]{Shaked2007}, and that the equilibrium distribution of log TFPQ has a decreasing hazard rate (DHR). These conditions ensure, by Theorem 3.B.20  in \cite{Shaked2007}, that the distribution of log TFPQ becomes more unequal in the sense of the dispersive order \citep[see Section 3.B in][]{Shaked2007}, which in turn implies an increase in the variance of log TFPQ. Both of these conditions are satisfied in our specification, because the exponential distribution has a constant hazard rate and a proportional improvement in log TFPQ implies an improvement in the hazard rate order sense for any distribution with a decreasing hazard rate.  

It remains an open question what general conditions would ensure that the equilibrium log TFPQ distribution is log-convex and that aggregate shocks improve said distribution in HRO sense.
\subsection{Numerical Exercise}\label{sec:numerical}

Next we move to a numerical illustration. Consider an environment in which aggregate fluctuations are driven exclusively by the reduced-form market efficiency shock $z_t$. The model is simulated by feeding in a stochastic process for $z_t$, which follows a two-state Markov chain taking values $\{0, z_h\}$, corresponding to boom and crisis states. The transition probabilities $(p_l, p_h)$ govern the persistence of each state, representing the probability of remaining in a boom or a crisis, respectively.

\paragraph{Calibration}

We calibrate the model at an annual frequency. Table \ref{table:1} reports the parameters fixed a priori, including standard macroeconomic values such as capital intensity $\alpha = 0.3$, labor intensity $\gamma = 0.6$, depreciation $\delta = 0.10$, discount factor $\beta = 0.96$, and CES parameter $\xi = 9$. The Markov transition probabilities $(p_l, p_h)$ are taken from \cite{khan2013} to match the durations of booms and crises. For simplicity, we set $\sigma_2 = 0$. As established earlier in this section, $\sigma_2$ is isomorphic to $\sigma_1$ in the computation for the second moments of our interests.

\begin{table} 
\centering %
\begin{tabular}{lcc } 
\hline \hline 
 Parameter& Description&Value \\
\hline 
 $\alpha$ & Capital intensity  & 0.3   \tabularnewline
$\gamma$ & Labor intensity  & 0.6   \tabularnewline 
$\delta$ &Depreciation rate  &  0.10  \tabularnewline
$\beta$ &Discount factor& 0.96   \tabularnewline
$\xi$ & CES elasticity of substitution  &  9 \tabularnewline
$p_l$ & Prob of staying in a boom  &  0.977 \tabularnewline
$p_h$ & Prob of staying in a recession    &0.688 \tabularnewline
$\sigma_2$ & Std of $\epsilon_2$  & 0 \tabularnewline
\hline \hline 
\end{tabular}\caption{Fixed Parameters}
\label{table:1} 
\end{table}

The remaining parameters are calibrated by minimizing the sum of squared deviations between model-implied moments and their empirical counterparts. The model is simulated for $10^4$ periods, and all moments are computed using the final 9,900 observations to eliminate the influence of initial conditions.

We target three classes of moments. First, we discipline the cross-sectional distributions of workers and firms by matching overall wage inequality---measured as the standard deviation of log wages---as well as the concentration of firm revenues, captured by the revenue shares of the top 10 percent of firms and of firms between the 50th and 90th percentiles. Second, we target the aggregate labor share, which in our framework reflects distortions arising from labor market wedges, as discussed in Section \ref{sec:focs}. Finally, we target the volatility of aggregate productivity, measured by the standard deviation of aggregate TFP. Wage inequality moments are taken from \cite{song2019}, firm revenue distribution moments from \cite{kwon2024}, and aggregate labor share and TFP from FRED. All empirical moments are computed over the 1978-2013 sample. Table~\ref{table:2} reports the calibration results: Panel A presents the calibrated parameter values, while Panel B compares the targeted moments with their model-implied counterparts.

\begin{table}
\centering %
\footnotesize
\resizebox{\textwidth}{!}{%
\begin{tabular}{lcc}  
\multicolumn{3}{c}{\textbf{Panel A: Calibrated Parameters}} \tabularnewline
\hline \hline 
 Parameter& Description&Value \\
\hline 
$\psi$ &  intensity of worker type $x$ & 0.4022   \tabularnewline 
$z_h$ & $z$ in a recession  &  0.3984 \tabularnewline
$\lambda_\theta$ & distribution of $\theta$ (type of firms)  & 2.6160 \tabularnewline
$\lambda_x$ & distribution of $x$ (type of workers)  & 0.8681 \tabularnewline
$\sigma_1$ & Std of $\epsilon_1$  & 0.2293 \tabularnewline
\hline \hline  \\
\multicolumn{3}{c}{\textbf{Panel B: Targets and Model Fit}} \tabularnewline
\hline \hline 
 Moment & Data & Model \\
\hline 
Labor share & 0.6097  & 0.6102   \tabularnewline
Wage inequality &  0.7666 &  0.7666  \tabularnewline
Share of revenue of top 10\% firms & 0.9074 & 0.8906   \tabularnewline
Share of revenue of top 50\%-top 10\% firms & 0.0842 &  0.0840 \tabularnewline
Std of TFP  & 0.0090 & 0.0090 \tabularnewline
\hline \hline 
\end{tabular}}\caption{Calibration and Model Fit }
  \caption*{\footnotesize\emph{Notes:} Labor share data and TFP data are retrieved from FRED.  
    Wage inequality is from \cite{song2019}.  Shares of revenue are from \cite{kwon2024}. All the moments are calculated over year 1978-2013.}
\label{table:2}
\end{table}

\paragraph{Impulse responses}
Figure \ref{fig:irf} reports the average impulse responses to an increase in the market-efficiency shock, $z_t = z_h$, which corresponds to a deterioration in allocative efficiency and hence a recessionary episode. The top-left panel shows that aggregate TFP declines on impact, reflecting the worsening allocation of labor across heterogeneous firms. Aggregate output falls by an even larger magnitude, driven both by lower TFP and by a gradual contraction in the capital stock.

\begin{figure}[H]
    \centering
    \includegraphics[
        width=0.85\textwidth,
        trim=1.5cm 8.cm 1.5cm 7cm,
        clip
    ]{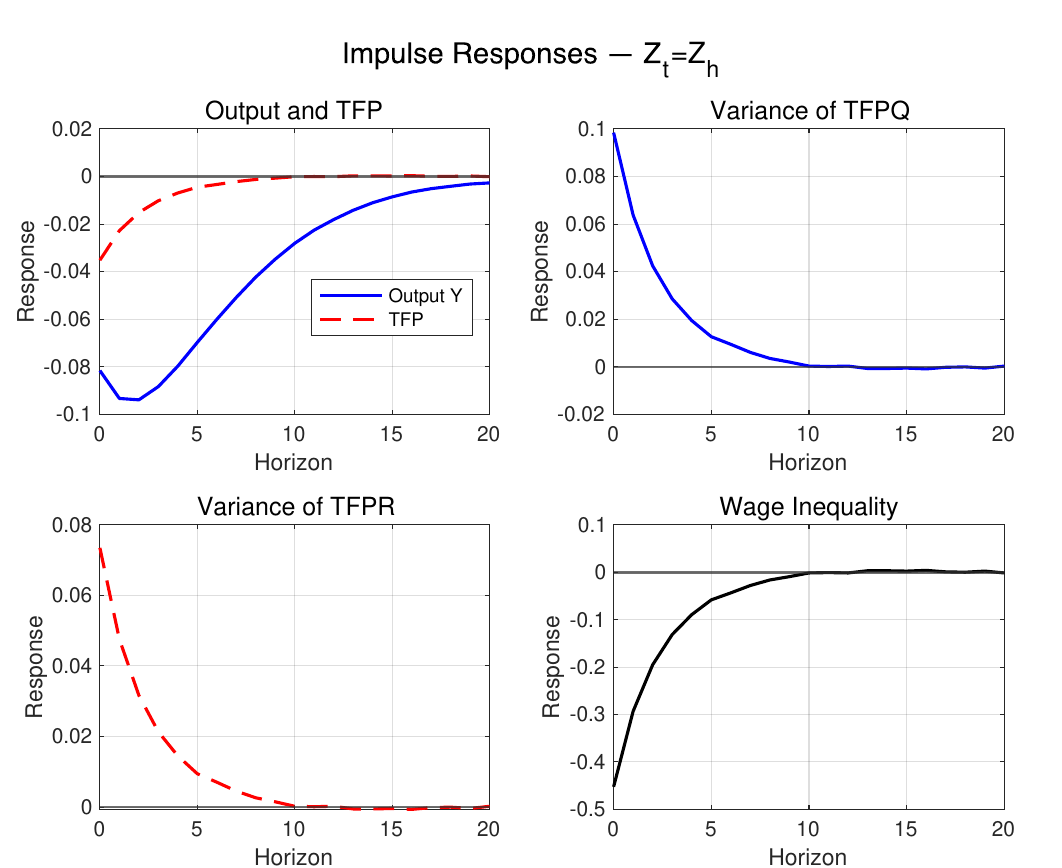}
    \caption{Impulse responses to a market-efficiency shock}
    \label{fig:irf}
\end{figure}

The remaining panels illustrate how the shock reshapes cross-sectional distributions. The variance of firm-level quantity productivity (log TFPQ) rises sharply on impact and remains elevated for several periods, indicating an increase in productivity dispersion. Intuitively, as market efficiency deteriorates, high-type firms contract disproportionately and concentrate employment on their most productive workers, amplifying differences in measured productivity across firms. A similar dynamic is observed for the dispersion of firm-level revenue productivity (log TFPR), which also increases before gradually converging back to its steady state.

In contrast, wage inequality declines persistently following the shock. As allocative efficiency worsens, high-productivity firms shrink and high-type positions become scarcer, flattening the wage schedule across worker types and compressing earnings dispersion. Taken together, the impulse responses highlight the central mechanism of the model: market-efficiency shocks generate opposing cyclical movements in productivity dispersion and wage inequality---productivity dispersion rises while wage dispersion falls during downturns---mirroring the key empirical regularities.

The magnitudes of the impulse responses are economically significant. The responses of GDP and aggregate TFP are reported in log differences. Following a simulated crisis---that is, a transition to the low-efficiency state---GDP contracts by more than 9 percent on impact, reflecting both a decline in allocative efficiency and a gradual reduction in the capital stock. The remaining panels report level deviations from their unconditional means. When $z_t=0$, the variances of log firm-level productivity (TFPQ) and log firm-level revenue (TFPR) are $0.1203$ and $0.0546$, respectively; when $z_t=z_h$, these variances increase to $0.2254$ and $0.1330$, indicating a substantial widening of productivity and revenue dispersion during crises. In contrast, wage inequality declines sharply, falling from $0.7901$ in booms to $0.3132$ in recessions.

While the quantitative responses are large relative to available empirical estimates,\footnote{See, for example, \cite{kehrig2015} and \cite{cunningham2023}.} they reflect the stylized nature of the experiment and serve to highlight the strength of the model's underlying mechanisms. 
When aggregate TFP fluctuations are driven jointly by the reduced-form market efficiency shock $z_t$ and the aggregate productivity shock $A_t$, the implied volatility of $z_t$ is smaller than in our numerical experiment. As implied by the analysis in Section \ref{sec:4.3},  the resulting fluctuations of other endogenous variables are closer in magnitude to those observed in the data.

\section{Structural Drivers of Business Cycles}\label{sec:4.3}
For simplicity and clarity in both exposition and notation, we have focused exclusively on aggregate fluctuations driven by time-varying $z_t$. As noted earlier, we do not model market frictions explicitly, so $z_t$ should be interpreted as a reduced-form shock capturing the combined influence of multiple structural shocks. The benefit of this approach is that it delivers closed-form implications for the distributions we care about.

That said, no two booms or recessions are necessarily alike, since they may arise from different underlying forces. Macroeconomists have proposed a range of structural shocks to account for business-cycle dynamics. We now examine the implications of several such shocks within our framework.
\paragraph{Productivity shock}
One of the most common shocks in business cycle models is the aggregate productivity shock. 
Recall that the production function is
\begin{align*}
Q_t(l, k, x, \theta)=A_t q(x, \theta) k^{\alpha}  l^\gamma,\label{production2}
\end{align*}
where $A_t$ denotes an aggregate productivity shock. Our earlier analysis shows that $\lambda_t$ is independent of $A_t$, and that $A_t$ does not affect the sorting equation (\ref{pam}). Consequently, in our framework, productivity shocks do not influence the dispersion of wage or productivity.

It is important to emphasize, however, that we have treated the micro-foundations of market distortions as a black box. In a model with explicit micro-founded frictions, firms of different types may respond heterogeneously to productivity shocks. Therefore, our result should not be interpreted as a general claim that productivity shocks never alter sorting patterns or dispersion moments. Rather, within the current framework, productivity shocks leave sorting and dispersion unchanged precisely because they do not generate any relative compositional shifts across firm types in the labor market. Without such shifts, the underlying worker-firm matching and the associated dispersion of wages or productivity remain unaffected.

\paragraph{Second-moment shocks}

Thus far, we have not imposed any specific structure or assumptions on the idiosyncratic process of firm type $\theta$. We have only assumed that $\theta\sim \mathrm{Exp}\left(\lambda_\theta\right)$. Consider the following simple idiosyncratic process
\begin{align*}
\theta_{j,t+1} =
\begin{cases}
\rho\,\theta_{j,t}, & \text{with probability } \rho,\\
\rho\,\theta_{j,t} + \varepsilon_{t+1}, & \text{with probability } 1-\rho,
\end{cases}
\end{align*}
where
$\varepsilon_{t+1} \sim \mathrm{Exp}(\lambda_\theta)$ is i.i.d  across firms and time. The process ensures  that $\theta$ follows an exponential distribution with fixed $\lambda_\theta$.

We can extend this framework to allow for a time-varying distribution of $\theta$ by considering the following process:
\begin{align*}
  \theta_{j,t+1} =
\begin{cases}
\rho\,\theta_{j,t}, & \text{with probability } p_{t+1},\\
\rho\,\theta_{j,t} + \varepsilon_{t+1}, & \text{with probability } 1-p_{t+1},
\end{cases}
\end{align*}
where
\begin{align*}
\varepsilon_{t+1} \sim \mathrm{Exp}(\lambda_{\theta,t+1}),
\qquad
p_{t+1} = \frac{\rho\,\lambda_{\theta,t+1}}{\lambda_{\theta,t}}.    
\end{align*}
For all $t$, $\lambda_{\theta t}\in\left\{\lambda^L_\theta,\lambda^H_\theta\right\}$ and follows a Markov process, with $\lambda_\theta^L<\lambda_\theta^H<\frac{\lambda_\theta^L}{\rho}$. This process guarantees that $\theta\sim \mathrm{Exp}\left(\lambda_{\theta t}\right)$ with a time-varying $\lambda_{\theta t}$ evolving according to a Markov process.

Note that the shock to $\lambda_\theta$ is not purely a second-moment shock: a decrease in $\lambda_\theta$ increases both the variance and the mean of $\theta$.

How does a fall in $\lambda_{\theta t}$ affect the distribution of jobs? Using equation (\ref{eqm:jobdist2}), we see that, similar to a decline in $z_t$, the decrease in $\lambda_{\theta t}$ reduces $\lambda_t$. From equation (\ref{eq:var_wage}), it then follows immediately that wage dispersion increases.

However, unlike the reduced-form shock in $z_t$, a decrease in $\lambda_\theta$ also increases the dispersions of firm-level TFPQ and TFPR. As shown in Appendix \ref{app:shocklambda}, this type of shock shifts wage inequality and productivity dispersion in the same direction.

\deferred[app:shocklambda]{\subsection{The effects of shocks to $\lambda_\theta$\label{app:shocklambda}}

Consider the equilibrium condition
\begin{equation}\label{eq:lambda_eq}
    -\lambda_t
    = b\!\left(\frac{\lambda_t}{\lambda_x}\right)^{\psi}
      - d z_t
      - \lambda_\theta,
\end{equation}
where all variables are strictly positive, $0<\psi<1$, $d>0$, and $z_t>0$. This equation implicitly defines $\lambda_t=\lambda_t(\lambda_\theta)$.

\bigskip
\noindent\textbf{Step 1. Comparative statics for $\left(\frac{\lambda_t}{\lambda_x}\right)^{2\psi}\lambda_\theta^{-2}$.}

Define
\[
A \equiv \left(\frac{\lambda_t}{\lambda_x}\right)^{2\psi}, \qquad B \equiv A \lambda_\theta^{-2}.
\]

By the implicit function theorem, differentiating \eqref{eq:lambda_eq} with respect to $\lambda_\theta$ gives
\[
\frac{\partial \lambda_t}{\partial \lambda_\theta}
= \frac{1}{1 + b \psi \lambda_x^{-\psi}\lambda_t^{\psi-1}} > 0.
\]

Taking logarithms of $B$ and differentiating:
\[
\frac{1}{B}\frac{\partial B}{\partial \lambda_\theta}
= 2\psi \frac{1}{\lambda_t} \frac{\partial \lambda_t}{\partial \lambda_\theta} - \frac{2}{\lambda_\theta}
= 2\left(
\frac{\psi}{\lambda_t(1 + b \psi \lambda_x^{-\psi}\lambda_t^{\psi-1})}
- \frac{1}{\lambda_\theta}
\right).
\]

From \eqref{eq:lambda_eq}, we can write
\[
\lambda_\theta = \lambda_t\big(1 + b \lambda_x^{-\psi} \lambda_t^{\psi-1}\big) - d z_t,
\]
so that
\[
\frac{\psi}{\lambda_t(1 + b \psi \lambda_x^{-\psi}\lambda_t^{\psi-1})} < \frac{1}{\lambda_\theta}.
\]

Hence
\[
\frac{\partial B}{\partial \lambda_\theta} < 0.
\]
This shows that $B=(\lambda_t/\lambda_x)^{2\psi}\lambda_\theta^{-2}$ increases when $\lambda_\theta$ decreases.

\bigskip
\noindent\textbf{Step 2. Comparative statics for $\big(e (\lambda_t/\lambda_x)^{2\psi} + m z_t\big)^2 \lambda_\theta^{-2}$.}

Here we have $e>0$ and $m>0$.
Define
\[
A \equiv \left(\frac{\lambda_t}{\lambda_x}\right)^{2\psi}, \qquad
S \equiv \big(e A + m z_t\big)^2 \lambda_\theta^{-2}.
\]

Take logarithms and differentiate with respect to $\lambda_\theta$:
\[
\frac{1}{S}\frac{\partial S}{\partial \lambda_\theta}
= 2 \frac{e\, \partial_{\lambda_\theta} A}{e A + m z_t} - \frac{2}{\lambda_\theta}.
\]

Since $A = (\lambda_t/\lambda_x)^{2\psi}$ and $\lambda_x$ is constant, the chain rule gives
\[
\partial_{\lambda_\theta} A
= 2\psi (\lambda_t)^{2\psi-1} \frac{1}{\lambda_x^{2\psi}} \frac{\partial \lambda_t}{\partial \lambda_\theta}
= 2 \psi \frac{A}{\lambda_t} \frac{\partial \lambda_t}{\partial \lambda_\theta}.
\]
Here we used the identity $A = (\lambda_t)^{2\psi}/(\lambda_x)^{2\psi}$.

\medskip
\noindent
From Step 1, we already have
\[
\frac{\partial \lambda_t}{\partial \lambda_\theta}
= \frac{1}{1 + b \psi \lambda_x^{-\psi}\lambda_t^{\psi-1}} > 0.
\]

\medskip
\noindent
Substituting into the derivative of $S$:
\[
\frac{e\, \partial_{\lambda_\theta} A}{e A + m z_t}
= \frac{2\psi e A}{(e A + m z_t) \lambda_t (1 + b \psi \lambda_x^{-\psi} \lambda_t^{\psi-1})} \le \frac{2\psi}{\lambda_t (1 + b \psi \lambda_x^{-\psi} \lambda_t^{\psi-1})}.
\]

Finally, using the inequality established in Step 1, we conclude
\[
\frac{\partial S}{\partial \lambda_\theta} < 0,
\]
so $S$ increases when $\lambda_\theta$ decreases.

\bigskip
\noindent\textbf{Conclusion.}  
Both quantities
\[
B = \left(\frac{\lambda_t}{\lambda_x}\right)^{2\psi} \lambda_\theta^{-2}, \quad
S = \big(e (\lambda_t/\lambda_x)^{2\psi} + m z_t\big)^2 \lambda_\theta^{-2}
\]
are strictly decreasing in $\lambda_\theta$. Section \ref{sec:4.3} shows that the TFPQ dispersion can be expressed as $B$, and the TFPR dispersion can be expressed as 
\[
\mathrm{Var}_t(\log TFPR_{jt})=S+(\frac{\eta^Q}{\xi})^2\left(\gamma^2\sigma_1^2+\alpha^2\sigma_2^2\right)
\]
if
\begin{align*}
e=1-\frac{1-\psi}{1+(1-\alpha-\gamma)(\xi-1)},
\qquad
m=\frac{\gamma}{1+(1-\alpha-\gamma)(\xi-1)}.     
\end{align*}
Therefore, a decrease in $\lambda_\theta$ increases both the dispersions of TFPQ and TFPR.

}

Another way to introduce second-moment shocks is through time-varying variances of  $\epsilon_{1}$ and $\epsilon_{2}$, the labor and capital wedges that are independent to firm types. For example, one can consider the following processes:
\begin{align*}
    \log\sigma_{1t}&=\rho_1\log\sigma_{1t-1}+\varepsilon_{1t}\\
    \log\sigma_{2t}&=\rho_2\log\sigma_{2t-1}+\varepsilon_{2t}
\end{align*}
where $\varepsilon_{1t}\sim N\left(0,\sigma_{l}^{2}\right)$, and $\varepsilon_{2t}\sim N\left(0,\sigma_{k}^{2}\right)$.

From equation (\ref{eqm:jobdist2}), it follows that the variances of $\epsilon_1$ and $\epsilon_2$---distortions uncorrelated with firm type $\theta$---do not affect $\lambda_t$. Consequently, shocks to these variances have no impact on the dispersion of wages or TFPQ, nor do they alter coefficients such as $\eta^Q$ or $\eta_{\theta t}^Q$. However, as indicated by equation (\ref{var_tfpr}), such shocks can influence the dispersion of TFPR: Increases in $\sigma_{1t}$ or $\sigma_{2t}$ raise the dispersion of revenue productivity across firms.

\section{Concluding Remarks}\label{sec:conclude}
This paper develops a tractable model in which heterogeneous firms and workers sort in the labor market, allowing productivity and wages to emerge endogenously from the allocation of talent across firms. We introduce a reduced-form market efficiency shock that affects high-type firms more strongly and show analytically how it generates opposing cyclical movements in wage and productivity dispersions---wage dispersion rises in booms while productivity dispersion narrows---mirroring key empirical patterns. We further explore how alternative structural shocks, including aggregate productivity and second-moment shocks, shape these distributions and their joint dynamics. Taken together, the framework provides a unified way to interpret business-cycle fluctuations in both wage inequality and productivity dispersion, reconciling several empirical regularities within a single, highly tractable theoretical framework.

A limitation of our approach is that market distortions enter in reduced-form through firm-specific wedges on labor and capital prices, while market efficiency shocks are modeled as time variation in these wedges. This modeling choice preserves analytical tractability and allows us to derive sharp, closed-form implications for the joint behavior of wage and productivity dispersions. An important direction for future work is to embed the same sorting mechanism into a fully micro-founded environment with explicit market frictions, structural shocks, and firm entry and exit. Such a framework would necessarily rely on numerical solution methods but would permit a richer quantitative assessment of how alternative frictions and shocks shape the co-movement of wages, productivity, and other variables over the business cycle.

The analysis in this paper abstracts from unemployment and matching frictions: all workers are employed, and matches form without search. Incorporating unemployment and search-and-matching frictions would again require numerical solutions, but doing so would open an important avenue for future work. In the current framework, the focus is naturally on wage inequality, which reflects dispersion in \textit{labor income} among \textit{employed} workers. While wage inequality is an important and widely studied object, introducing unemployment would allow us to study an additional---and empirically crucial---dimension of income inequality. Recent evidence suggests that income inequality is countercyclical, with unemployment playing a central role.\footnote{See, for example, \cite{heathcote2020}.} 
Understanding the sources of disparity between wage inequality and income inequality is a promising direction for future research.

 \bibliographystyle{chicago}

\bibliography{references}

\appendix

\newpage


\setcounter{figure}{0}
\renewcommand{\thefigure}{A.\arabic{figure}}

\setcounter{table}{0}
\renewcommand{\thetable}{B.\arabic{table}}

\setcounter{section}{0}
\setcounter{subsection}{0}

\setcounter{footnote}{0}

\begin{center}
    \Large\textbf{Online Appendix} \\
   \vspace{0.6em}
   \large\textit{}
   \vspace{1.5em}
\end{center}

\section{Omitted Proofs and Derivations}\label{app: proofs}
\shownow{lem: firmsolution}

\shownow{lem:jobdist}

\shownow{unknown}

\shownow{lem:lambda}

\shownow{app:shocklambda}

\end{document}